%% file: main.tex
\documentclass[acmsmall]{acmart}

\usepackage{bussproofs}
\usepackage{cleveref}
\usepackage{makecell}
\usepackage{paralist}
\usepackage{listings}
\usepackage{multirow}
\usepackage{xspace}
\usepackage{subcaption}
\usepackage{thmtools}

\AtBeginDocument{%
  }

\setcopyright{cc}
\setcctype{by}
\acmDOI{10.1145/3808341}
\acmYear{2026}
\acmJournal{PACMPL}
\acmVolume{10}
\acmNumber{PLDI}
\acmArticle{263}
\acmMonth{6}
\acmSubmissionID{pldi26main-p807-p}
\received{2025-11-14}
\received[accepted]{2026-04-03}

\begin{document}

\input{macros}

\theoremstyle{remark}
\newtheorem{remark}{Remark}

\title{Simplifying Safety Proofs with Forward-Backward Reasoning and Prophecy}

\author{Eden Frenkel}
\orcid{0009-0009-4589-2173}
\affiliation{%
  \institution{Tel Aviv University}
  \city{Tel Aviv}
  \country{Israel}
}
\email{edenfrenkel@mail.tau.ac.il}

\author{Kenneth L. McMillan}
\orcid{0009-0000-9380-1939}
\affiliation{%
  \institution{University of Texas at Austin}
  \city{Austin}
  \country{USA}
}
\email{kenmcm@cs.utexas.edu}

\author{Oded Padon}
\orcid{0009-0006-4209-1635}
\affiliation{%
  \institution{Weizmann Institute of Science}
  \city{Rehovot}
  \country{Israel}
}
\email{oded.padon@weizmann.ac.il}

\author{Sharon Shoham}
\orcid{0000-0002-7226-3526}
\affiliation{%
  \institution{Tel Aviv University}
  \city{Tel Aviv}
  \country{Israel}
}
\email{sharon.shoham@gmail.com}

\begin{abstract}
We propose an incremental approach for safety proofs that decomposes a proof with a complex inductive invariant into a sequence of simpler proof steps. Our proof system combines rules for (i)~forward reasoning using inductive invariants, (ii)~backward reasoning using inductive invariants of a time-reversed system, and (iii)~prophecy steps that add witnesses for existentially quantified properties. We prove each rule sound and give a construction that recovers a single safe inductive invariant from an incremental proof. The construction of the invariant demonstrates the increased complexity of a single inductive invariant compared to the invariant formulas used in an incremental proof, which may have simpler Boolean structures and fewer quantifiers and quantifier alternations. Under natural restrictions on the available invariant formulas, each proof rule strictly increases proof power. That is, each rule allows to prove more safety problems with the same set of formulas. Thus, the incremental approach is able to reduce the search space of invariant formulas needed to prove safety of a given system. A case study on Paxos, several of its variants, and Raft demonstrates that forward-backward steps can remove complex Boolean structure while prophecy eliminates quantifiers and quantifier alternations.
\end{abstract}




\maketitle

\def\defaultHypSeparation{\ } %

\input{0-intro}

\input{2-overview}

\input{1-preliminaries}

\input{3-incremental}

\input{4-fwd-bwd}

\input{5-witnesses}

\input{6-eval}

\input{7-related}

\section*{Data Availability}
The data and code supporting this paper are available in the paper's artifact~\cite{artifact} and as part of the \href{https://github.com/wilcoxjay/mypyvy}{\texttt{mypyvy} tool}~\cite{mypyvy}.

\begin{acks}
We thank James R. Wilcox for developing the \texttt{mypyvy} tool and for his assistance in integrating our proof checker into it.
We thank Raz Lotan and the anonymous reviewers for their helpful and insightful comments.
The research leading to these results has received funding from the
European Research Council under the European Union's Horizon 2020 research and innovation programme (grant agreement No [759102-SVIS]).
This research was partially supported by Israel Science Foundation (ISF) grant No.\ 2117/23.
This research was partially supported by Israel Science Foundation research grant (ISF's No.\ 4136/25) and the Maimonides Fund's Future Scientists Center,
a research grant from the Center for New Scientists at the Weizmann Institute of Science, and a grant from the Azrieli Foundation.

\end{acks}

\ifarxiv
\appendix
\crefalias{section}{appendix}

\input{8-proof-appendix}
\fi

\bibliographystyle{ACM-Reference-Format}
\bibliography{refs}

\end{document}

%% file: macros.tex
\newif\ifarxiv
\arxivtrue
\newif\ifproofs
\proofstrue

\newcommand{\nats}{\mathbb{N}}

\newcommand{\states}{\mathbb{S}}
\newcommand{\family}{\mathcal{F}}
\newcommand{\powerset}{\mathcal{P}}
\newcommand{\reach}{\mathcal{R}}

\newcommand{\seq}[1]{{\langle #1 \rangle}}
\newcommand{\fbinv}{\operatorname{Inv}}

\newcommand{\false}{\textbf{false}}
\newcommand{\true}{\textbf{true}}

\newcommand{\lang}{\mathcal{L}}
\newcommand{\wildquant}{{*}}

\newcommand{\predicates}{\mathbb{P}}
\newcommand{\fwd}[2]{\mathbf{F}^{#1}_{#2}}
\newcommand{\fwdbwd}[2]{\mathbf{FB}^{#1}_{#2}}
\newcommand{\fwdbwdproph}[2]{\mathbf{FBP}^{#1}_{#2}}

\newcommand{\Inv}{\operatorname{Inv}}
\newcommand{\FInv}{\Inv^\rightarrow}
\newcommand{\FBInv}{\Inv^\rightleftharpoons}
\newcommand{\FBPInv}{\Inv_\mathrm{pro}^\rightleftharpoons}
\newcommand{\Tr}{T}
\newcommand{\fwdstep}{\mathbf{fwd}}
\newcommand{\bwdstep}{\mathbf{bwd}}

\newcommand{\Fsys}{\mathbf{F}}
\newcommand{\FIsys}{\mathbf{FI}}
\newcommand{\FBIsys}{\mathbf{FBI}}
\newcommand{\FBPIsys}{\mathbf{FBPI}}

\renewcommand{\phi}{\varphi}

\newcommand{\sort}[1]{\mathsf{#1}}
\newcommand{\snode}{\sort{node}}
\newcommand{\squorum}{\sort{quorum}}
\newcommand{\sround}{\sort{round}}
\newcommand{\svalue}{\sort{value}}
\newcommand{\relation}[1]{{\mathit{#1}}}
\newcommand{\rmember}{\relation{member}}
\newcommand{\ronea}{\relation{start\_round\_msg}}
\newcommand{\ronebnovote}{\relation{join\_msg\_no\_vote}}
\newcommand{\ronebmaxvote}{\relation{join\_msg}}
\newcommand{\rproposal}{\relation{propose\_msg}}
\newcommand{\rvote}{\relation{vote\_msg}}
\newcommand{\rcurrentround}{\relation{current\_round}}
\newcommand{\rdecision}{\relation{decision}}
\newcommand{\action}[1]{{\textsc{#1}}}
\newcommand{\asendonea}{\action{start\_round}}
\newcommand{\ajoinround}{\action{join\_round}}
\newcommand{\apropose}{\action{propose}}
\newcommand{\acastvote}{\action{vote}}
\newcommand{\adecide}{\action{learn}}
\newcommand{\none}{{\bot}}
\newcommand{\theory}{{\Gamma}}
\newcommand{\theorytotalorder}{{\theory_\text{total order}}}

\newcommand{\updr}{PDR$^\forall$\xspace}

%% file: 0-intro.tex
\section{Introduction}

A prominent approach for verifying safety properties of complex systems is based on inductive invariants. 
Inductive invariants are properties that hold in the initial states and are preserved by every step of the system, and ergo they hold in all reachable states of the system. 
Therefore, a safety property can be proved by finding an inductive invariant that implies it (a \emph{safe inductive invariant}).
As systems and properties become more complex, the inductive invariants required to establish their safety may become extremely complex. 
In particular, if we use formulas in first-order logic to model systems, their safety properties, and their inductive invariants,
the complexity of inductive invariants is exhibited in the Boolean structure of the corresponding formulas, the number of quantifiers, and the structure of quantifier alternations. 
The complexity of inductive invariants poses a major challenge for verification. First, complex inductive invariants induce a huge search space, making it difficult to synthesize them automatically. Second, even validating that a formula is indeed an inductive invariant becomes challenging, both because of scale and because of the effect of quantifier structure on the complexity and decidability of validating inductiveness. %

One common methodology for simplifying safety proofs relies on simplifying the components being verified. This is the general idea behind modular reasoning, which divides the verification task to sub-tasks that are simpler to reason about such that the safety proof of the overall system follows from the proofs of the sub-tasks. Examples of this methodology include assume-guarantee reasoning~\cite{DBLP:phd/ethos/Jones81,DBLP:conf/ifip/Jones83,DBLP:journals/toplas/AbadiL93,DBLP:journals/toplas/AbadiL95,pnueli-modular,DBLP:conf/cav/HenzingerQR98} and Owicki-Gries proofs~\cite{DBLP:journals/acta/OwickiG76}, where each sub-task verifies a component of the system. However, there are cases where even a single component requires complicated proofs. In these cases, simplifications that are tied to the modular structure of the system do not apply.

In this paper we investigate a complementary approach for simplifying safety proofs which, rather than dividing the system being verified, divides the proof itself into simpler sub-proofs. We do so by introducing \emph{incremental proofs} that combine \emph{forward} and  \emph{backward} reasoning and a suitable form of \emph{prophecy} variables.

Our starting point
is the well-known observation that safety can be established by \emph{incremental induction},
where in each step an invariant is proved by induction while assuming previously proved invariants. Similarly to safe inductive invariants, all the invariants proven in such an incremental proof hold on all reachable states. Our first observation is that we can allow incremental proofs that combine forward and backward steps, where a backward step uses an invariant that holds in all states that are backward reachable from the bad states.
Therefore, the combination of forward and backward steps allows the proof to use predicates that do not hold on either the forward reachable or the backward reachable states, but maintains the fact that the predicates hold on all paths from an initial state to a bad state.
We observe that this combination can reduce the complexity (e.g., Boolean structure) of predicates used in proofs %
when compared to safe inductive invariants. 

Next, we observe that we can sometimes simplify the kind of quantification needed in forward or backward invariants by introducing prophecy witnesses for some
existentially
quantified property, and using the witnesses in subsequent proof steps instead of
quantified variables.
Our notion of prophecy is similar to prophecy variables used in prior work, e.g.~\cite{prophic3,prophecy-made-simple,VickMcMillan}, but our development includes two key novelties.
First, our notion of prophecy admits a sound and complete characterization in terms of an auxiliary safety problem: checking whether a certain prophecy witness can be added amounts to checking the safety of an auxiliary transition system.
Second, our notion of prophecy combines particularly well with forward and backward proofs.
Prophecy may not only decrease the number of quantifiers needed
in proofs,
but also the quantification depth and possibly eliminate quantifier alternations.

Based on these observations we define a proof system for incremental safety proofs that allows to combine forward-backward reasoning and prophecy in a synergetic way.
This is enabled by our characterization of the soundness of adding prophecy witnesses as a safety problem.
For example, it is possible to use forward-backward proof steps for identifying sound prophecy witnesses, and then use these witnesses in subsequent  proof steps. This kind of interaction between the 
proof steps 
is key to simplifying the formulas used in proofs. 

We demonstrate the ability of incremental proofs to simplify safety proofs both theoretically and via a case study.
We prove that if we restrict the set of formulas that can be used as inductive invariants in proofs, then incremental proofs have an increased proof power compared to safe inductive invariants. %
That is, there are cases where no safe inductive invariant exists within the chosen set of formulas but an incremental proof does exist. 
Moreover, we show that each proof rule strictly increases the overall proof power of the proof system.
Reducing the complexity of formulas used in proofs
may be beneficial for proof automation.

As a case study we consider the Paxos protocol for consensus~\cite{paxos,paxos-made-simple}, several of its variants, and Raft~\cite{raft}. Safety verification of
these protocols has been extensively investigated in recent years due to their importance in distributed systems and their intricate safety proofs that require a non-clausal Boolean structure and quantifier alternations~\cite{ironfleet,weaken,pfolic3,DuoAI,Swiss,DBLP:journals/pacmpl/PadonLSS17,DBLP:conf/osdi/ZhangHKCP24,DBLP:conf/osdi/ZhangSCKP25}.
To expand our case study, for each variant of Paxos we consider
several versions differing in how the protocol and its safety property are modeled. We show that in some cases we can significantly simplify the safety proofs of these protocols using our proof system, for example reducing a proof that required quantifier alternations and non-clausal predicates to clauses that are purely universally quantified.
The case study also demonstrates that forward-backward reasoning can aid in finding useful prophecy witnesses that might be difficult to find using forward reasoning only. It therefore illustrates the significance of the synergetic combination of forward-backward and prophecy provided by our proof system. To the best of our knowledge, this kind of interaction between forward-backward reasoning and prophecy is novel. 

In summary, this paper makes the following contributions:
\begin{itemize}
    \item We introduce a proof system for incremental safety proofs that combines forward reasoning, backward reasoning and introduction of prophecy in a synergetic way.
    
    \item We provide a characterization of the soundness of our notion of prophecy witnesses using an auxiliary safety problem.
    
    \item We relate incremental safety proofs to the standard notion of safe inductive invariants via a translation of incremental proofs to inductive invariants. 
    
    \item We investigate the proof power of incremental proofs and show that they may sometimes eliminate Boolean connectives as well as quantifiers, including nested and alternating quantifiers, that are required in safe inductive invariants. The interaction of forward-backward reasoning and prophecy is key to achieving such simplifications.
    
    \item We present a case study where we examine 
    the Paxos consensus protocol, several of its variants, and Raft,
    and demonstrate how incremental proofs simplify their safety proofs.  
\end{itemize}

The rest of the paper is organized as follows:
\Cref{sec:overview} presents an overview of our approach for simplifying proofs,
\Cref{sec:prelim} includes preliminary definitions and notation, 
\Cref{sec:incremental-proofs} describes our proof system for incremental forward-backward safety proofs,
\Cref{sec:proph} extends the proof system with a proof rule for introducing prophecy witnesses,
and \Cref{sec:case-study} details the case study demonstrating proof simplification for Paxos variants and Raft.
Finally, \Cref{sec:related} discusses related work and concludes the paper.

%% file: 2-overview.tex
\section{Overview}
\label{sec:overview}

In this section we illustrate and motivate the different proof rules comprising our proof system for incremental safety proofs. We do so by presenting several examples, where incremental proofs allow us to use  invariants that are simpler than the safe inductive invariants required for establishing safety of the original problem. We demonstrate simplifications both in the Boolean structure of invariants and in their quantification.

\subsection{Simplifying the Boolean Structure of Invariants}
\label{sec:overview-boolean}

Our first motivating example considers a simple token-passing protocol between a \emph{dealer} and two \emph{players}, modeled with propositional variables: %
$a, d, p_1, p_2$. When $a$ is true we say that the protocol is \emph{active}, and this property of the protocol never changes; $d$, $p_1$, and $p_2$ represent whether the dealer and the players each hold a token, respectively.

The protocol starts with the second player not holding a token, thus the initial condition is:
\[ \iota  = \neg p_2. \]
Then, in each step, one of three transitions can take place, each of which is captured by a formula that relates the values of the variables $a, d, p_1, p_2$ before the transition to their values after the transition, encoded by $a', d', p_1', p_2'$. If the protocol is active, a token can be passed from the dealer to the first player:
\[ \tau_1 = a \wedge a' \wedge d\wedge \neg d' \wedge p_1' \wedge (p_2' \leftrightarrow p_2) \]
Moreover, the first player can always pass a token to the second:
\[ \tau_2 = (a' \leftrightarrow a) \wedge 
(d' \leftrightarrow d) \wedge p_1 \wedge \neg p_1' \wedge p_2' \]
Lastly, if the protocol is inactive, the dealer can gain tokens arbitrarily:
\[ \tau_3 = \neg a \wedge \neg a' \wedge d' \wedge (p_1' \leftrightarrow p_1) \wedge (p_2' \leftrightarrow p_2). \]
All in all, the transition formula for this protocol is given by $\tau = \tau_1\vee \tau_2\vee \tau_3$.

Say we would like to verify that it is impossible for the dealer and both players to all hold tokens simultaneously. That is, we would like to prove that the following bad states are unreachable:
\[\beta = d\wedge p_1 \wedge p_2. \]
Unfortunately, the property $\neg \beta \equiv \neg d \vee \neg p_1 \vee \neg p_2$ is not preserved by $\tau_3$, and is therefore not inductive. In fact, there is no clause that can serve as a safe inductive invariant, since any clause that implies $\neg \beta$ must be a subset of the literals $\{\neg d, \neg p_1, \neg p_2\}$, and it is easy to check that none produce an inductive invariant. For a similar reason, there is no negation of a clause, or conjunction of literals, that is a safe inductive invariant for this safety problem (such a conjunction would have to be exactly $\neg p_2$ to be implied by $\iota$, which is not inductive).

Thus, when looking for a safe inductive invariant for this simple protocol, either manually or automatically, we need to consider candidates containing both disjunctions and conjunctions.
One such safe inductive invariant for this problem is as follows:
\[ \varphi = (a \wedge \neg d) \vee \neg p_1 \vee \neg p_2, \]
which can be read as ``either the protocol is active and the dealer does not hold a token, or one of the two players does not hold a token''. This property holds on all initial states and is preserved by the transitions: $\tau_1$ precisely results in a state where the protocol is active and the dealer has no token, $\tau_2$ results in a state where the first player does not hold a token, and $\tau_3$ considers an inactive protocol and does not change $p_1$ and $p_2$, which means that $\neg p_1 \vee \neg p_2$ must hold both before and after the transition. It also holds that $\varphi \Rightarrow \neg \beta$, which makes the inductive invariant safe.

When rewritten in conjunctive normal form, $\varphi$ is equivalent to the conjunction of two clauses: $\varphi_1 = a \vee \neg p_1 \vee \neg p_2$ and $\varphi_2 = \neg d \vee \neg p_1 \vee \neg p_2$. We observe that it is sometimes possible to prove conjunctive invariants such as $\varphi$ \emph{incrementally}, by proving
simpler auxiliary inductive invariants until those, when taken together, manage to prove safety. This can be done by iteratively verifying a sequence of inductive invariants $\varphi_1,\ldots,\varphi_n$, where for every $i=1,...,n$ we add $\bigwedge_{j<i} \varphi_j$ as additional axioms of the safety problem when verifying $\varphi_i$. Namely, we strengthen the initial states to
$\iota\wedge\bigwedge_{j<i} \varphi_j$
and the transitions to
$\bigwedge_{j<i} \varphi_j \wedge\tau \wedge \bigwedge_{j<i} \varphi'_j$.
If, in addition, $\bigwedge_{i\leq n}\varphi_i$ implies safety, then $\bigwedge_{i\leq n}\varphi_i$ is a safe inductive invariant of the original safety problem.

For our simple example such a proof can be constructed from $\varphi_1$ and $\varphi_2$, since $\varphi_1$ is inductive by itself, and $\varphi_2$ is inductive under the assumption that $\varphi_1$ holds on every pre-state and post-state of a transitions, i.e.,
$\varphi_1\wedge \varphi_2\wedge\tau\wedge\varphi_1' \Rightarrow \varphi'_2$. Since $\varphi_1\wedge \varphi_2$ implies $\neg\beta$ (and is in fact equivalent to $\varphi$), this forms an incremental forward proof of safety, where in each step we added a clause of size 3, whereas no safe inductive invariant of the problem can be clausal in form.

However, we can construct a different kind of incremental proof that uses even simpler auxiliary invariants. This is done by observing that instead of a proof method that accumulates inductive invariants, which are properties of reachable states, we can accumulate properties of states that appear on error traces, i.e., on traces from an initial state to a bad state
(even when such traces are hypothetical and do not exist).
This allows us to prove safety ``by contradiction'':
if the accumulated properties imply $\neg \beta$ or $\neg \iota$, we deduce that there are no error traces and the safety problem is safe.

To that end, in each step of the incremental proof we allow introducing either an inductive invariant, or a \emph{backward-inductive invariant}, i.e., a formula $\varphi$ such that $\beta \Rightarrow \varphi$ and $\varphi' \wedge \tau \Rightarrow \varphi$. Such a formula holds on all backward-reachable states from $\beta$, and is therefore an invariant of all error traces of the safety problem.
Intuitively, this allows us to use simpler inductive invariants in incremental proofs, because in subsequent proof steps we can rely on auxiliary invariants that might not be true of reachable states, whereas in incremental forward proofs all auxiliary invariants were actual (forward) invariants of the problem.

For the simple dealer-player protocol we consider, we can use a backward-inductive invariant $\varphi_3 = \neg a \vee d$, since it holds on all bad states and is preserved by backward transitions (in bad states the dealer holds a token, and can lose a token in a backward transition only if the protocol is inactive). Then, using this assumption as an axiom for a forward proof step, the transition $\tau_1$ is disabled (all its post-states violate $\varphi_3$), and therefore $\varphi_4 = \neg p_1 \vee \neg p_2$ is a (forward) inductive invariant. Thus, in this incremental forward-backward proof we can conclude that all states on traces from initial to bad states satisfy $\varphi_3 \wedge\varphi_4 = (\neg a \vee d) \wedge (\neg p_1 \vee \neg p_2)$, which implies $\neg \beta = \neg (d\wedge p_1 \wedge p_2)$. In this way we have shown the protocol to be safe using the auxiliary invariants $\neg a \vee d$ and $\neg p_1\vee \neg p_2$, compared to $a \vee \neg p_1 \vee \neg p_2$ and $\neg d \vee \neg p_1 \vee \neg p_2$ in the incremental forward proof. This may seem like a modest simplification in the auxiliary invariants one needs to consider, but this stems from the very simple nature of the protocol. As we shall see in the sequel, the presence of quantification can significantly transform the power of forward-backward proofs compared to forward ones.

Note that $\varphi_3 \wedge \varphi_4$ is not an inductive invariant of the original protocol: it does not hold in the initial states and it is not preserved by $\tau_1$. It is only guaranteed to hold in all the states along traces from an initial state to a bad state. However, the incremental proof does induce (another) safe inductive invariant for the protocol.
In \Cref{sec:incremental-forward-backward-proofs} we show that any incremental forward-backward proof can be converted into a safe inductive invariant of the problem which consists of a Boolean combination of the predicates used in the proof (\Cref{thm:FB-invariant}). For the proof using the backward invariant $\varphi_3 = \neg a \vee d$ and the forward invariant $\varphi_4 = \neg p_1 \vee \neg p_2$, this results in $\neg \varphi_3 \vee \varphi_4 = \neg (\neg a \vee d) \vee (\neg p_1 \vee \neg p_2)$, which is equivalent to the safe inductive invariant $\varphi = (a \wedge \neg d) \vee (\neg p_1 \vee \neg p_2)$ presented before.

\subsection{Forward-Backward Proofs and Quantification}
\label{sec:overview-quant}

We introduce quantification and first-order semantics into our motivating protocol by %
replacing the propositional variables with unary relations $a(\cdot), d(\cdot), p_1(\cdot), p_2(\cdot)$. 
Instead of a single dealer and two players, the protocol works with \emph{teams}, each of which is composed of a dealer and two players, where each of $d(t)$, $p_1(t)$, and $p_2(t)$ is true if the dealer, first player, or second player of team $t$ holds a token, respectively.
Similarly, a team $t$ is active if $a(t)$ holds. Initially, all second players hold no tokens:
\[ \iota  = \forall x. \neg p_2(x). \]
Similarly to before, the first transition allows a dealer from an active team to pass a token to a first player, but not necessarily from his own team. For readability,
we use TLA+ style notation where we specify transitions as formulas $[\varphi]_S$ where $S$ denotes the modified set of values (i.e., anything outside $S$ is not modified by the transition):
\[ \tau_1 = \exists x, y. [a(x) \wedge d(x) \wedge \neg d'(x) \wedge p_1'(y)]_{\{d(x),p_1(y)\}}. \]
The second transition lets a first player of a team pass a token to the second player of the team:
\[ \tau_2 = \exists x. [p_1(x) \wedge \neg p_1'(x) \wedge p_2'(x)]_{\{p_1(x),p_2(x)\}}. \]
And the last transition allows a dealer of an inactive team to gain a token:
\[ \tau_3 = \exists x. [\neg a(x) \wedge d'(x)]_{\{d(x)\}} \]
The bad states of this updated protocol are the states where all dealers hold a token, and there is a team where both players hold a token as well:
\[ \beta = \forall x. d(x) \wedge \exists y. p_1(y) \wedge p_2(y). \]

As in \Cref{sec:overview-boolean}, $\neg\beta$ is not preserved by $\tau_3$ and is therefore not inductive. One safe inductive invariant for the problem is
\[ \varphi = (\exists x. a(x) \wedge \neg d(x)) \vee (\forall x. \neg p_1(x) \vee \neg p_2(x)). \]
Intuitively, this inductive invariant says that ``Either there is an active team where the dealer has no token, or in every team there is a player that does not hold a token''. Recall that in \Cref{sec:overview-boolean}, we had no clausal safe inductive invariant, but could find an incremental forward proof that used only clausal auxiliary invariants, namely $a \vee \neg p_1 \vee \neg p_2$ and $\neg d \vee \neg p_1 \vee \neg p_2$. An analogous proof, however, seems impossible here, since the invariant cannot even be decomposed into conjuncts in a similar way. A similar forward-backward proof, however, is possible: we can first take $\varphi_1 = \forall x. \neg a(x) \vee d(x)$, which is backward-inductive. Using it as an axiom once again disables $\tau_1$, which makes $\varphi_2 = \forall x. \neg p_1(x) \vee \neg p_2(x)$ inductive in the forward direction. Taken together, $\varphi_1 \wedge\varphi_2 \Rightarrow\neg \beta$, which concludes the incremental forward-backward safety proof.

Converting this forward-backward proof into a safe inductive invariant in the way described in \Cref{sec:incremental-forward-backward-proofs} results in the invariant $\neg \varphi_1 \vee \varphi_2 = \neg (\forall x. \neg a(x) \vee d(x)) \vee (\forall x. \neg p_1(x) \vee \neg p_2(x))$, which is once again equivalent to $\varphi$. The incremental proof itself only needed to consider, and to perform inductiveness checks, over universally quantified clauses. On the other hand, this problem does not have any safe inductive invariant, or an incremental proof that is purely forward or purely backward, that consists of universally quantified clauses.
This claim can be proved by the \updr algorithm~\cite{updr} which is able to provide a proof that no universally quantified safe inductive invariant exists for this problem (and its time-reversal), let alone a clausal one. Indeed, we use the implementation of \updr in \texttt{mypyvy}~\cite{mypyvy} to obtain such a proof.

\subsection{Using Prophecy to Reduce Quantifier Nesting and Alternations}
\label{sec:overview-proph}

In the example of \Cref{sec:overview-quant}, we showed that an incremental forward-backward proof reduces the number of quantifiers needed in auxiliary invariants compared to those in a safe inductive invariant or a forward incremental proof. However, because a forward-backward proof can always be converted into a safe inductive invariant which is a Boolean combination of its auxiliary invariants (\Cref{thm:FB-invariant}), forward-backward proofs on their own cannot reduce quantifier nesting or alternations.
To that end, we propose using a notion of prophecy defined in \Cref{sec:proph} which replaces quantified variables with explicit witnesses.

For an example where quantifier alternations are needed in the inductive invariant, and are eliminated by the use of prophecy, we use a slightly simplified protocol compared to \Cref{sec:overview-boolean,sec:overview-quant}, which nonetheless demonstrates how we gain quantifier nesting and alternations from using prophecy in incremental forward-backward proofs. This protocol still consists of an \emph{active} relation $a(\cdot)$, as well as a dealer relation, which is now binary, $d(\cdot,\cdot)$. As before, the \emph{active} relation is unchanging and denotes whether a dealer of a certain team can pass tokens. Now, however, the dealer of each team can hold a token for every team, i.e., $d(t_1, t_2)$ indicates that the dealer of $t_1$ still holds the token associated with $t_2$. The players are no longer part of the modeled protocol.

Initially there is at least one inactive team, and dealers hold all tokens:
\[ \iota = \exists x. \neg a(x) \wedge \forall x, y. d(x, y). \]
There is only one transition, where a dealer of an active team can lose any of its tokens (by passing it to the players of the appropriate team, which are not modeled):
\[ \tau = \exists x, y. [a(x) \wedge \neg d'(x, y)]_{\{d(x,y)\}}. \]
Lastly, the bad states are those where all dealers have lost at least one token:
\[ \beta = \forall x. \exists y. \neg d(x, y). \]
Intuitively, this state cannot be reached because all reachable states contain an inactive team whose dealer always holds all of their tokens.
However, like before, $\neg \beta$ is not preserved by $\tau$ and is therefore not inductive, but
we can write a safe inductive invariant for this problem that precisely captures the inactive team:
\[ \varphi = \exists x. \neg a(x)\wedge \forall y. d(x, y). \]
This invariant contains quantifier nesting and alternations.
Since a safe inductive invariant for this problem would sensibly require quantifier nesting, we cannot have a forward-backward proof using auxiliary invariants with a single quantifier as before, since such a proof can be converted into nesting-free safe inductive invariant (\Cref{thm:FB-invariant}).
Additionally, once again we obtain a proof that this problem has no universal forward proof or a universal backward proof via the \updr algorithm~\cite{updr} implemented in \texttt{mypyvy}~\cite{mypyvy}, clausal or otherwise.

To allow incremental proofs that avoid quantifier nesting and alternations, as well as simplify the Boolean structure of the required predicates, we propose using \emph{prophecy}. The idea is to allow a proof step that adds a fresh symbol $w$ to the safety problem, along with an axiom $\varphi$ that prophesies some property of $w$ along every trace. This is a generalization of the previous proof steps, which added axioms that did not mention any new symbols. Moreover, similarly to how in a forward-backward proof step we required $\varphi$ to be an invariant of error traces (to ensure soundness), with the prophecy proof rule we require that adding $\varphi$ as an axiom is guaranteed to preserve all error traces of the problem. That is, a prophecy step $\varphi$ is \emph{sound} for a safety problem $\Pi = (\iota,\tau,\beta)$ if every error trace of $\Pi$ can be transformed (by assigning a certain value to $w$) into an error trace of the problem with the additional axiom, i.e., the safety problem with initial states $\iota\wedge \varphi(w)$, transitions $\varphi(w)\wedge\tau\wedge(w' = w)\wedge (\varphi(w))'$, and bad states $\beta\wedge \varphi(w)$.

In the above example, any forward trace starting from $\iota$ must have some team $x$ such that $\neg a(x)$, since this holds in all initial states and $a(\cdot)$ remains unchanged.
Formally, this can be verified by checking that $\iota\Rightarrow \exists x. \neg a(x)$, and that $\neg a(x)$, which is a formula with a free variable $x$, is forward-preserved, i.e., $\neg a(x) \wedge \tau \Rightarrow \neg a'(x)$.
Since these verification conditions are valid, on forward traces from $\iota$ it is sound to assume that there exists some witness $w$ which satisfies $\neg a(x)$. Next, we can add $w$ as an unchanging witness to the safety problem along with the axiom $\varphi_1 = \neg a(w)$, resulting in an equivalently safe safety problem.

Under this prophecy assumption, the transition $\tau$ becomes disabled for $x=w$, and so the formula $\varphi_2 = \exists y. \neg d(w, y)$ is backward inductive. Taken together $\varphi_1\wedge\varphi_2 = \neg a(w) \wedge \exists y. \neg d(w, y)$ implies $\neg \iota$, which means that the problem is safe. Note that this proof used only existentially quantified auxiliary invariants, with a single quantified variable and no Boolean connectives. However, we proved a property whose safe inductive invariant contained quantifier nesting and alternations.

In \Cref{sec:proph} we generalize this proof method by precisely defining the notion of a \emph{sound prophecy}
used in this paper
and showing that the soundness of a given prophecy can itself be stated as a safety problem to be proven, which fits neatly into our approach of incrementally proving various safety problems until we are able to prove the original one in question.
Moreover, simiarly to before, we describe how a forward-backward incremental proof with prophecy can be converted into a safe inductive invariant for the problem (\Cref{thm:FBP-invariant}). Like with forward-backward proofs, this conversion can create complex Boolean combinations of the auxiliary invariants used in the proof, but may also introduce quantifiers to eliminate the use of the witness symbols used in prophecy. For example, converting the above
proof
results in the safe inductive invariant $\varphi = \exists x. \neg a(x)\wedge \forall y. d(x, y)$, where in this simple example, the witness has been effectively existentially quantified.

%% file: 1-preliminaries.tex
\section{Preliminaries}
\label{sec:prelim}

In this section we provide the background necessary for our developments. 
We use first-order logic to model systems and their properties. System states are encoded using a first-order vocabulary $\Sigma$ which consists of constant symbols, function symbols and relation symbols. Given a first-order vocabulary $\Sigma$, a \emph{transition system} is a pair $(\iota, \tau)$, where $\iota$ is a closed formula over $\Sigma$ representing initial states of the system and $\tau$ is a closed formula over $\Sigma \uplus \Sigma'$ that represents the transitions of the system. $\Sigma' =\{a' \mid a \in \Sigma\}$ is a copy of the vocabulary used to represent the post-state of a transition. For a formula $\varphi$ over $\Sigma$ we denote by $\varphi'$ the formula obtained by substituting each symbol in $\Sigma$ with its counterpart in $\Sigma'$. We use $\tau^{-1}$ to denote the inverse of $\tau$ obtained by replacing each symbol from $\Sigma$ by its counterpart in $\Sigma'$ and vice versa.

A \emph{safety problem} is a triple $(\iota, \tau,\beta)$, where $(\iota,\tau)$ is a transition system and $\beta$ is a closed formula over the same vocabulary $\Sigma$ representing bad states. The safety problem $(\iota,\tau,\beta)$ is safe if there is no sequence of $\tau$-transitions (a trace of $(\iota,\tau)$) from an initial state to a bad state. 

A closed formula $\varphi$ over $\Sigma$ is an \emph{inductive invariant} for a transition system $(\iota,\tau)$ if $\iota\Rightarrow \varphi$ (initiation) and $\varphi \wedge \tau \Rightarrow \varphi'$  (consecution), which means that $\varphi$ holds in all the initial states and it is preserved by the transitions of the system. The formula $\varphi$ is a \emph{safe inductive invariant} for the safety problem $(\iota,\tau,\beta)$ if, in addition, $\varphi \Rightarrow \neg \beta$ (safety). An inductive invariant over-approximates the set of states that can be reached from the initial states of the system. Therefore, if it
is safe,
it implies that the reachable states are disjoint from the bad states, thus the safety problem is safe.

%% file: 3-incremental.tex
\section{Incremental Proofs}
\label{sec:incremental-proofs}

In this paper we present a family of proof systems for constructing incremental safety proofs.
\ifproofs\else
The soundness of each proof system follows from its sound inference rules (see~\cite{forward-backward-extended}). \fi
In this section we start with incremental  proofs that combine forward and backward reasoning.

\subsection{Incremental Forward Proofs}
\label{sec:incremental-forward-proofs}

The statements of the proof systems we present are safety problems. 
Each proof system is defined by a set of inference rules whose premises and conclusions are safety problems, possibly with additional side conditions.
Provable statements correspond to safety problems that are provably safe using the inference rules of the system.
A proof of a safety problem $\Pi$ is then a tree where each node is a safety problem, the root is $\Pi$ and each node is obtained by applying an inference rule.

Since we are interested in analyzing the proof power of each system,
we sometimes parameterize a proof system by the class of 
predicates $\predicates$
used as inductive
invariants in proof rules.
We compare how expressive that class has to be to prove the same safety problems in different proof systems.

The most basic proof component of the incremental proofs we consider is proving an invariant by induction, i.e., by showing that all initial states satisfy it, and that all transitions preserve it.
Since not all invariants of a transition system are inductive, to make such a system (relatively) complete, we also include an inference rule that allows strengthening the bad states (i.e., weakening the proven invariant), which preserves safety.

\begin{definition}
    \label{def:fwd-ind-system}
    The \emph{forward proof system} $\Fsys$ consists of the following proof rules:
    \begin{center}
        \AxiomC{}
        \LeftLabel{(Ind)}
        \RightLabel{$\iota\Rightarrow \varphi$, $\varphi\wedge\tau \Rightarrow \varphi'$}
        \UnaryInfC{$(\iota, \tau, \neg\varphi)$}
        \bottomAlignProof
        \DisplayProof
        \hspace{1.5em}
        \AxiomC{$(\iota, \tau, \neg \varphi)$}
        \LeftLabel{(Cons)}
        \RightLabel{$\varphi\Rightarrow \neg \beta$}
        \UnaryInfC{$(\iota, \tau, \beta)$}
        \bottomAlignProof
        \DisplayProof
    \end{center}
    For a class of predicates $\predicates$, we denote by $\Fsys^\predicates$ the system $\Fsys$ where applications of (Ind) are restricted to $\varphi\in\predicates$.
\end{definition}

\begin{theorem}[Soundness]
    \label{thm:F-sound}
    Every safety problem that has an $\Fsys$-proof is safe.
\end{theorem}

\ifproofs
\begin{proof}
    All inference rules of $\Fsys$ are sound.
\end{proof}
\fi

Note that
all $\Fsys$-proofs can be collapsed into the following form, by collapsing consecutive consequence rules:
\begin{prooftree}
    \AxiomC{}
    \LeftLabel{(Ind)}
    \RightLabel{$\iota\Rightarrow \varphi$, $\varphi\wedge\tau \Rightarrow \varphi'$}
    \UnaryInfC{$(\iota, \tau, \neg \varphi)$}
    \LeftLabel{(Cons)}
    \RightLabel{$\varphi \Rightarrow \neg \beta$}
    \UnaryInfC{$(\iota, \tau, \beta)$}
\end{prooftree}

Thus the $\Fsys$ proof system has no meaningful way to simplify proofs by dividing them into multiple applications of the induction rule. 
To allow such simplifications we introduce incremental proofs. Incremental proofs may use an additional proof rule that allows adding previously proved invariants as axioms of the safety problem. That is, in order to prove a safety problem $\Pi = (\iota, \tau, \beta)$ safe, the incremental proof rule requires proving that
\begin{inparaenum}
    \item a predicate $\varphi$ is an invariant with respect to $\iota$ and $\tau$, and
    \item the problem $\Pi$, when restricted to states that satisfy $\varphi$, is safe.
\end{inparaenum}

\begin{definition}
    \label{def:fwd-inc-ind-system}
    The \emph{incremental forward proof system} $\FIsys$ consists of the proof rules of $\Fsys$ (\Cref{def:fwd-ind-system}), and the additional incremental proof rule:
    \begin{center}
        \AxiomC{$(\iota, \tau, \neg\varphi)$}
        \AxiomC{$(\iota\wedge\varphi, \tau\wedge\varphi\wedge\varphi', \beta\wedge\varphi)$}
        \LeftLabel{(Inc)}
        \BinaryInfC{$(\iota, \tau, \beta)$}
        \bottomAlignProof
        \DisplayProof
    \end{center}
    For a class of predicates $\predicates$, we denote by $\FIsys^\predicates$ the system $\FIsys$ where applications of (Ind) %
    are restricted to  $\varphi\in\predicates$,
    and by $\FIsys^\predicates_n$ the system where proofs use at most $n$ applications of (Ind). 
\end{definition}

\begin{theorem}[Soundness]
    \label{thm:FI-sound}
    Every safety problem that has an $\FIsys$-proof is safe.
\end{theorem}

\ifproofs
\begin{proof}
    All inference rules of $\FIsys$ are sound.
\end{proof}
\fi

Like the rest of the proof systems presented in this paper, an incremental $\FIsys$-proof of a safety problem can always be converted into a safe inductive invariant for the same problem.
This is another way to prove soundness theorems like \Cref{thm:F-sound,thm:FI-sound}, but can also be used to reason about how expressive inductive invariants need to be to match the proof power of a given proof system. In the case of incremental forward proofs $\FIsys^\predicates$, each proof can be converted into a safe inductive invariant in the form of a conjunction over $\predicates$, where the number of conjuncts precisely corresponds to the number of (Ind) rules used in the proof. This conversion is given below.

\begin{definition}
    \label{def:F-invariant}
    For every proof $P\in \FIsys$ define $\FInv(P)$ by induction on proof structure:
    \begin{itemize}
        \item if the root of $P$ is a conclusion of (Ind), denoted $(\iota,\tau,\neg\varphi)$, then $\FInv(P)= \varphi$;
        \item if the root of $P$ is a conclusion of (Cons) with a premise whose proof is $\Tilde{P}$, then
        \[\FInv(P)= \FInv(\Tilde{P});\]
        \item if the root of $P$ is a conclusion of (Inc) with premises whose proofs are $P_1$ and $P_2$, then
        \[ \FInv(P)= \FInv(P_1)\wedge \FInv(P_2). \]
    \end{itemize}
\end{definition}

\begin{theorem}
    \label{thm:F-invariant}
    Let $P\in \FIsys^\predicates$ be an incremental forward proof of a safety problem $\Pi$. Then $\FInv(P)$ is a safe inductive invariant of $\Pi$, and if $P$ contains $n\in \nats$ instances of the (Ind) rule, then $\FInv(P)$ is a conjunction of $n$ predicates from $\predicates$.
\end{theorem}

\Cref{thm:F-invariant}, whose proof is provided in
\ifarxiv
\Cref{sec:appendix}%
\else
the extended version of the paper~\cite{forward-backward-extended}%
\fi, ensures that any proof in $\FIsys^\predicates_n$ can be converted into a safe inductive invariant
in $\wedge_n \predicates$
for the same safety problem, where $\wedge_n \predicates$ denotes the class of predicates that includes conjunctions of at most $n$ predicates from $\predicates$.
This hints that incremental proofs may avoid conjunctions that are needed in a safe inductive invariant. In fact, there exist safety problems whose safe inductive invariant requires $n$ conjuncts that can be proven incrementally. These problems are provable in $\FIsys^\predicates_n$ but not in $\FIsys^\predicates_{n-1}$ (which subsumes $\Fsys^\predicates$). 
These results about the proof power of $\FIsys$ are summarized in \Cref{thm:FI-power}.

\begin{theorem}
    \label{thm:FI-power}
    For every $\predicates$, all  safety problems that are provable in $\FIsys_n^\predicates$ are also provable in $\Fsys^{\wedge_n\predicates}$. 
    However, for every $n>0$ there exist a set of predicates $\predicates$ and a safety problem $\Pi$ such that $\Pi$ is provable in $\FIsys_n^\predicates$ but not provable in $\FIsys_{n-1}^\predicates$.
\end{theorem}

\begin{corollary}
    If $\predicates$ is closed under conjunction, then any $\FIsys^\predicates$-provable problem is $\Fsys^\predicates$-provable.
\end{corollary}

%% file: 4-fwd-bwd.tex
\subsection{Incremental Forward-Backward Proofs}
\label{sec:incremental-forward-backward-proofs}

The next development in our proof system comes from the observation that a safety problem is safe if and only if its time-reversed problem is safe. That is, to prove the safety of $\Pi=(\iota, \tau, \beta)$, we can equivalently prove the safety of $\Pi^{-1}=(\beta, \tau^{-1}, \iota)$. 
Intuitively, time-reversal allows us to prove safety problems whose
safe inductive invariants require negations, since $\varphi$ is a safe inductive invariant of $\Pi^{-1}$ if and only if $\neg \varphi$ is a safe inductive invariant of $\Pi$.
Although this observation seems simple, the key insight is that when combined with incrementality, a time-reversing inference step, which allows to switch back and forth between forward reasoning and backward reasoning, results in a significant increase of proof power. Namely, since incremental proofs allow us to mimic conjunctive invariants, the addition of negation through time-reversal induces a rich Boolean structure for inductive invariants that contain both negations and conjunctions (which form a functionally complete set of connectives). This means that by using incremental proofs that allow time-reversal steps we may sometimes simplify the Boolean structure of
auxiliary inductive invariants
used in the proof in cases where a safe inductive invariant of the safety problem requires a complex Boolean structure (beyond conjunctions).

\begin{definition}
    \label{def:fwd-bwd-inc-system}
    The \emph{incremental forward-backward proof system} $\FBIsys$ consists of the proof rules of $\FIsys$ (\Cref{def:fwd-ind-system}), and the additional time-reversal proof rule:
    \begin{center}
        \AxiomC{$(\beta, \tau^{-1}, \iota)$}
        \LeftLabel{(Rev)}
        \UnaryInfC{$(\iota, \tau, \beta)$}
        \bottomAlignProof
        \DisplayProof
    \end{center}
    For a class of predicates $\predicates$, we denote by $\FBIsys^\predicates$ the system $\FBIsys$ where applications of (Ind) are restricted to $\varphi\in\predicates$,
    and by $\FBIsys^\predicates_n$ the system where proofs use at most $n$ applications of (Ind).
\end{definition}

\begin{theorem}[Soundness]
    \label{thm:FBI-sound}
    Every safety problem that has an $\FBIsys$-proof is safe.
\end{theorem}

\ifproofs
\begin{proof}
    All inference rules of $\FBIsys$ are sound.
\end{proof}
\fi

Although (Rev) is the only necessary addition to the proof system for harnessing backward induction, it is often more concise to write proofs using derived backward variants of (Ind) and (Inc), which directly allow backward induction and incrementality.

\begin{theorem}
    \label{def:FI-derived}
    The following are derived rules of $\FBIsys$:
    \begin{center}
        \AxiomC{}
        \LeftLabel{(B-Ind)}
        \RightLabel{$\beta \Rightarrow \varphi$, $\varphi' \wedge \tau \Rightarrow \varphi$}
        \UnaryInfC{$(\neg\varphi, \tau, \beta)$}
        \bottomAlignProof
        \DisplayProof
        \hspace{1.5em}
        \AxiomC{$(\neg \varphi, \tau, \beta)$}
        \LeftLabel{(B-Cons)}
        \RightLabel{$\varphi\Rightarrow \neg \iota$}
        \UnaryInfC{$(\iota, \tau, \beta)$}
        \bottomAlignProof
        \DisplayProof\\[0.1cm]
        \AxiomC{$(\neg\varphi, \tau, \beta)$}
        \AxiomC{$(\iota\wedge\varphi, \tau\wedge\varphi\wedge\varphi', \beta\wedge\varphi)$}
        \LeftLabel{(B-Inc)}
        \BinaryInfC{$(\iota, \tau, \beta)$}
        \bottomAlignProof
        \DisplayProof
    \end{center}
\end{theorem}

\begin{proof}
    Each of the rules can be derived by applying (Rev) to each of the premises, applying the respective $\FIsys$ rule, and applying (Rev) to the resulting conclusion.
\end{proof}

Similarly to \Cref{def:F-invariant} for $\FIsys$, every $\FBIsys$-proof can be converted into a safe inductive invariant for the safety problem in its root. As mentioned before, the only change is the handling of (Rev), which adds a negation to the computed invariant due to the time-reversal. When combined with the conjunctions introduced by the (Inc) rule in \Cref{def:F-invariant}, the computed invariant can now be equivalent to any arbitrary Boolean combination predicates from the considered class.

\begin{definition}
    \label{def:FB-invariant}
    For every proof $P\in \FBIsys$ define $\FBInv(P)$ by induction on proof structure:
    \begin{itemize}
        \item if the root of $P$ is not a conclusion of (Rev), define $\FBInv(P)$ like $\FInv(P)$ in \Cref{def:F-invariant};
        \item if the root of $P$ is a conclusion of (Rev) with a premise whose proof is $\Tilde{P}$, then
        \[\FBInv(P)= \neg\FBInv(\Tilde{P}).\]
    \end{itemize}
\end{definition}

\begin{theorem}
    \label{thm:FB-invariant}
    Let $P\in \FBIsys^\predicates$ be an incremental forward-backward proof of a safety problem $\Pi$. Then $\FBInv(P)$ is a safe inductive invariant of $\Pi$, and if $P$ contains $n\in \nats$ instances of the (Ind) rule, then $\FBInv(P)$ is a Boolean combination of $n$ predicates from $\predicates$.
\end{theorem}

The proof of the theorem is deferred to 
\ifarxiv
\Cref{sec:appendix}%
\else
the extended version of the paper~\cite{forward-backward-extended}%
\fi.

The overall results about the proof power of $\FBIsys$ are summarized as follows.

\begin{theorem}
    \label{thm:FB-power}
    For every $\predicates$, all safety problems that are provable in $\FBIsys_n^\predicates$ are also provable in $\Fsys^{\predicates_n}$, where $\predicates_n$ consists of all Boolean combinations of at most $n$ predicates from $\predicates$.
    However, for every $n>0$ there exist a set of predicates $\predicates$ and a safety problem $\Pi$ such that $\Pi$ is provable in $\FBIsys_n^\predicates$ but
    not provable in either $\FBIsys_{n-1}^\predicates$, $\FIsys^\predicates$, or $\Fsys^\predicates$.
\end{theorem}

\begin{proof}
    The first part of the theorem is a direct corollary of \Cref{thm:FB-invariant}.
    For the second part,
    define the vocabulary with $n$ nullary relations $\Sigma = \{p_1,\dots,p_n\}$ and the set of predicates $\predicates = \{p_1,\dots,p_n\}$.
    Denote $\varphi = p_1\wedge (\neg p_2 \vee (p_3 \wedge (\neg p_4 \vee \cdots)))$ and consider the safety problem $\Pi = (\varphi, \bot, \neg\varphi)$. This safety problem is provable in $\FBIsys^\predicates_n$, but
    not provable in either $\FBIsys_{n-1}^\predicates$, $\FIsys^\predicates$, or $\Fsys^\predicates$, since any safe inductive invariant of $\Pi$ must be equivalent to $\varphi$, and there exists no such formula consisting of fewer than $n$ predicates, or purely of conjunctions over $\predicates$.
\end{proof}

%% file: 5-witnesses.tex
\section{Incremental Forward-Backward Proofs with Prophecy}
\label{sec:proph}

The proof systems presented in \Cref{sec:incremental-proofs} let us prove safety properties using inductive invariants of a simpler Boolean structure than that of any safe inductive invariant of the original safety problem. However, the simplification afforded by $\FIsys$ and $\FBIsys$ \emph{only} allows simplifying the Boolean structure of inductive invariants; that is, if a safe inductive invariant cannot be expressed as a Boolean combination of several simpler predicates, e.g., due to quantification, then there is no hope for these proof systems to simplify the kind of predicates involved in a safety proof.

Our goal in this section is to augment the forward-backward proof system $\FBIsys$ with an additional rule that is able to eliminate quantifier nesting and alternation, and thus simplify proofs of problems that require quantified invariants that cannot otherwise be broken down into a Boolean combination of simpler predicates. We first introduce a general prophecy proof rule and then discuss specific heuristics for adding prophecy.

\subsection{Prophecy in Forward-Backward Proofs}
\label{sec:general-prophecy}

We introduce a \emph{prophecy proof rule} that prophesies
that there exists an element $x$ that always satisfies some formula $\varphi(x)$
and uses it to simplify the proof by extending the vocabulary with a fresh constant symbol, $w$, that serves as a \emph{witness} for the existential quantifier. The rule then adds $\varphi(w)$ as an axiom. To ensure soundness, it is required that the added axiom does not affect the error traces of the system. Then, by reasoning about the safety problem with the prophecy transformation, we are able to prove properties of the witness which translate to quantified properties of the original system. Since this proof system is an extension of $\FBIsys$, it is able to eliminate both quantification as well as logical connectives such as conjunctions and disjunctions.

We begin by describing the safety problem resulting from adding the prophecy witness constant. For the sake of simplicity, we consider adding a single witness which remains fixed along any given trace.
The approach naturally generalizes to multiple prophecy symbols, and to a prophecy that is updated in each transition according to a given update rule.

\begin{definition}
    \label{def:prophecy-refinement}
    Let $\Pi = (\iota,\tau,\beta)$ be a safety problem over vocabulary $\Sigma$, let $w$ be a fresh constant symbol, and let $\varphi(x)$ be a formula over $\Sigma$.
    Define the safety problem $\Pi^w_{\varphi}$ over $\Sigma\cup \{w\}$ as
    \[ \Pi^w_{\varphi} = (\iota\wedge\varphi(w), \varphi(w) \wedge \tau \wedge w' = w \wedge (\varphi(w))', \beta\wedge\varphi(w)). \]
    Then $\varphi(w)$ is a \emph{sound prophecy} for $\Pi$ if
    for every trace of $\Pi$ from $\iota$ to $\beta$, there exists an interpretation of $w$ that results in a trace of $\Pi^w_\varphi$ from $\iota\wedge\varphi(w)$ to $\beta\wedge\varphi(w)$.
\end{definition}

The soundness condition states that there is a suitable witness for every trace from an initial state to a bad state, where the witness satisfies $\phi$ all along the trace. This means that, when a bad state is reached, there exists an $x$ (the witness) such that $\phi(x)$ has been always true in the past. Therefore, we can formulate the soundness condition in a complete way as a temporal safety problem whose bad states are bad states of the original system where the aforementioned condition does not hold. Namely, for every $x$, either $\phi(x)$ does not hold at the moment or it has not been always true in the past. To track whether $\phi(x)$ has always been true in the past, we follow a standard tableau construction for the ``always in the past'' temporal operator, and augment the state of the system with a fresh predicate $m(\cdot)$ that remembers for each $x$ whether $\phi(x)$ has been always true in the past. Initially, every $x$ that satisfies $\phi(x)$ is added to $m$, and in every transition, if $x$ continues to satisfy $\phi(x)$, it remains in $m$. This ensures that if an element is not in $m$ then it has not always been true in the past. That is, we can use $\neg m(x)$ to identify elements $x$ that have not always been true in the past. The resulting safety problem is defined in the following theorem.

\begin{theorem}
    \label{thm:sound-prophecy-as-safety}
    Let $\Pi = (\iota,\tau,\beta)$ be a safety problem over vocabulary $\Sigma$, let $w$ be a fresh constant symbol, and let $\varphi(x)$ be a formula over $\Sigma$.
    Then $\varphi(w)$ is a sound prophecy for $\Pi$ 
    if and only if the following safety problem, over vocabulary $\Sigma \cup \{m(\cdot)\}$, where $m(\cdot)$ is a fresh relation symbol, is safe:
    \[\Pi_{\varphi}^\mathrm{sound} = (\iota\wedge \forall x. \varphi(x) \to m(x),
    \tau \wedge \forall x. (m(x)\wedge \varphi(x)\wedge \varphi'(x)) \to m'(x), \beta\wedge \forall x. \varphi(x) \to \neg m(x)).\]
    We call a safe inductive invariant of $\Pi_{\varphi}^\mathrm{sound}$ a \emph{soundness invariant} for $\varphi(w)$.
\end{theorem}

Next, we present the additional inference rule that allows proving a safety problem by first proving that a given prophecy $\varphi(w)$ is sound, and also proving that the safety problem under this prophecy assumption, and with a witness symbol, is safe. The soundness of this inference rule is a direct consequence of \Cref{def:prophecy-refinement}, which implies that for a sound prophecy $\varphi(w)$, a safety problem $\Pi$ is safe if and only if $\Pi^w_\varphi$ is safe; and \Cref{thm:sound-prophecy-as-safety}, which guarantees that the premise $\Pi^\mathrm{sound}_\varphi$ ensures that the prophecy used in the inference rule is sound.

\begin{definition}
    \label{def:fwd-bwd-proph-system}
    The \emph{incremental forward-backward proof system with prophecy} $\FBPIsys$ consists of the proof rules of $\FBIsys$ (\Cref{def:fwd-bwd-inc-system}), and the additional prophecy proof rule below, where $\Pi$ is a safety problem over vocabulary $\Sigma$, $w$ is a fresh constant symbol, and $\varphi(x)$ is a formula over $\Sigma$:
    \begin{center}
        \AxiomC{$\Pi_\varphi^\mathrm{sound}$}
        \AxiomC{$\Pi_\varphi^w$}
        \LeftLabel{(Proph)}
        \BinaryInfC{$\Pi$}
        \bottomAlignProof
        \DisplayProof
    \end{center}
    For a class of predicates $\predicates$, we denote by $\FBPIsys^\predicates$ the system $\FBPIsys$ where applications of (Ind) are restricted to $\varphi\in\predicates$,
    and by $\FBPIsys^\predicates_n$ the system where proofs use at most $n$ applications of (Ind).
\end{definition}

\begin{theorem}[Soundness]
    \label{thm:FBIP-sound}
    Every safety problem that has an $\FBPIsys$-proof is safe.
\end{theorem}

\ifproofs
\begin{proof}
    All inference rules of $\FBPIsys$ are sound.
\end{proof}
\fi

Similarly to how we derived the rules (B-Ind), (B-Cons), and (B-Inc) of $\FBIsys$ in \Cref{def:FI-derived}, it is possible to derive a backward (Proph) rule by time-reversing each of the premises before (Proph) and then time-reversing its conclusion.
However, 
the resulting rule is isomorphic to the original (Proph) rule
(with the polarity of $m$ flipped), and is therefore redundant.

Next, we describe how to convert an incremental proof with prophecy into a safe inductive invariant for the same safety problem. As before, the (Ind) rule determines the predicates used in the proof, whereas (Inc) creates conjunctions, and (Rev) creates negations.
The crucial change is that when handling the (Proph) proof rule, we take the safe inductive invariant of $\Pi^\mathrm{sound}_\varphi$, denoted $\xi$, and substitute every occurrence of $m(t)$ for some term $t$ with $\psi[t/w]$, where $\psi$ is the safe inductive invariant of $\Pi^w_\varphi$.
We denote the result of this operation by $\xi[\psi/m]$, and assume w.l.o.g. that $\xi$ and $\psi$ use distinct variables.
Intuitively, this results in an inductive invariant for $\Pi$ because, since $\psi$ is an inductive invariant for witness elements that always satisfy $\varphi(x)$, $\psi[t/w]$ can be used as a valid interpretation of $m(t)$ to substitute into the soundness invariant $\xi$ while maintaining inductiveness.
Moreover, $\xi[\psi/m]$ implies the safety of $\Pi$. This is true
because the soundness invariant $\xi$ implies that any bad state we reach must have an element that satisfies both $\varphi(x)$ and $m(\cdot)$. Therefore,
the resulting invariant $\xi[\psi/m]$
implies that any bad state we reach must have a witness $w$ that satisfies $\varphi(w)$ and $\psi$. Moreover, $\psi$ implies the safety property of $\Pi^w_\varphi$, which states that if we have a witness $w$ satisfying $\varphi(w)$, then we are not in a bad state. This
shows we cannot reach a bad state at all, since if we have a witness $w$ satisfying both $\varphi(w)$ and $\psi$, we are not in a bad state.

Typically, the soundness invariant $\xi$ would include existential quantifiers since it has to imply the safety property of $\Pi^\mathrm{sound}_\varphi$, which is $\beta \to \exists x. \varphi(x) \wedge m(x)$. Moreover, occurrences of $m(\cdot)$ in $\xi$ are likely to be applied to terms $t(x)$ that include these existentially quantified variables. Therefore, when the occurrences of $m(t(x))$ in $\xi$ are substituted by $\psi[t(x)/w]$, the quantifiers of $\psi$ become nested under the quantifiers of $\xi$. Since the constant $w$ in $\psi$ is substituted by $t(x)$, this is often a real nesting that cannot be broken into a Boolean combination of un-nested quantifiers. The reason is that 
$\psi$ may include constraints that mix the quantified variables of $\psi$ and $w$, in which case
$\psi[t(x)/w]$ includes constraints that mix the quantified variables of $\psi$ and $x$.
Thus, the use of the (Proph) proof rule may allow us to reduce the number of quantifiers, and even quantifier nesting, that is required in a safe inductive invariant.

\begin{definition}
    \label{def:FBP-invariant}
    For every proof $P\in \FBPIsys$ define $\FBPInv(P)$ by induction on proof structure:
    \begin{itemize}
        \item if the root of  $P$ is not a conclusion of (Proph), define $\FBPInv(P)$ like $\FBInv(P)$ in \Cref{def:FB-invariant};
        \item if the root of $P$ is a conclusion of (Proph) with premises $\Pi_\varphi^\mathrm{sound}$ and $\Pi_\varphi^w$, whose proofs are $P_\mathrm{sound}$ and $P_w$, respectively,
        let $\xi = \FBPInv(P_\mathrm{sound})$ and $\psi = \FBPInv(P_w)$.
        Then, $\FBPInv(P) = \xi[\psi/m]$, where the notation indicates that we substitute every application $m(t)$ in $\xi$, for any term $t$, with the formula $\psi[t/w]$.
    \end{itemize}
\end{definition}

\begin{theorem}
    \label{thm:FBP-invariant}
    Let $P\in \FBPIsys^\predicates$ be an incremental forward-backward proof of a safety problem $\Pi$. Then $\FBPInv(P)$ is a safe inductive invariant of $\Pi$, and if $P$ contains $n\in \nats$ instances of the (Ind) rule, and the quantification depth in predicates of $\predicates$ is at most $d$, then $\FBPInv(P)$ is a formula with quantification depth at most $n\cdot d$.
\end{theorem}

The crux of the proof of  \Cref{thm:FBP-invariant} is 
formalized and proved
\ifarxiv
in \Cref{lem:prophecy-invariant-conversion} in \Cref{sec:appendix}%
\else
in the extended version~\cite{forward-backward-extended}%
\fi. 
Next, consider how the additional quantification depth in computed invariants may arise. To compute the invariant, in each application of the (Proph) inference rule we take an invariant $\xi$ of the premise $\Pi_\varphi^\mathrm{sound}$ and an invariant $\psi$ of the premise $\Pi_\varphi^w$,
and replace all occurrences of the relation $m(\cdot)$ applied to a term $t$ in $\xi$ with $\psi[t/w]$. Thus, if $\xi$ has quantification depth $k$, and $\psi$ has quantification depth $\ell$, the resulting formula $\xi[\psi/m]$ has quantification depth of at most $k + \ell$. Meanwhile, the quantification depth of the invariant computed for the (Inc) rule is the maximum of the quantification depths associated with the invariants of its two premises, and the rest of the inference rules have no effect on the quantification depth in the computed invariant at all. Therefore, if $n$ predicates are involved in the computation of the invariant, and each has quantification depth of at most $d$, then the resulting invariant cannot exceed quantification depth $n\cdot d$.

\begin{example}
    \label{ex:proph-quantifiers}
    Recall the problem $\Pi = (\iota, \tau, \beta)$ with vocabulary $\Sigma = \{a(\cdot), d(\cdot, \cdot)\}$ from \Cref{sec:overview-proph}:
    \begin{align*}    
        \iota &= \exists x. \neg a(x)\wedge \forall x,y. d(x,y)  \qquad \qquad  \beta = \forall x.\exists y. \neg d(x,y) \\
        \tau &= (\forall x. a(x) \leftrightarrow a'(x)) \wedge 
        \exists x_0, y_0. a(x_0) \wedge \forall x,y. d'(x,y) \leftrightarrow (d(x,y) \wedge (x\neq x_0 \vee y\neq y_0)) %
    \end{align*}
    Initial states of this system all contain some \emph{inactive} element, i.e., an element not satisfying $a(\cdot)$, and the interpretation of $a(\cdot)$ remains fixed along transitions. Moreover, in initial states all pairs of elements satisfy $d(\cdot, \cdot)$, but in each transition a pair $(x_0, y_0)$ is removed from the relation, assuming $a(x_0)$ holds. The bad states of the system are those where all elements $x$ have some $y$ that is not in the relation with them. Intuitively, this cannot happen, since for any trace there exists some $x$ where $\neg a(x)$ holds throughout, and for this $x$ no pairs can be removed from the relation, which initially contains $(x, y)$ for every $y$. This can be proven using the following inductive invariant with quantification depth 2 (with a real nesting of quantifiers): $\exists x. \neg a(x) \wedge \forall y. d(x,y)$.

    However, using the proof system $\FBPIsys$, the safety of this system can be proven using predicates with a single existential quantifier only, by introducing a prophecy witness for $\neg a(x)$, where we let $\Pi^\mathrm{sound}_{\neg a(x)} = (\iota_m, \tau_m, \beta_m)$ and $\Pi^w_{\neg a(x)} = (\iota_w, \tau_w, \beta_w)$:
    \begin{prooftree}
        \AxiomC{}
        \LeftLabel{(Ind)}
        \UnaryInfC{$(\iota_m, \tau_m, \neg \exists x. \neg a(x)\wedge m(x))$}
        \LeftLabel{(Cons)}
        \UnaryInfC{$\Pi^\mathrm{sound}_{\neg a(x)}$}
        \AxiomC{}
        \LeftLabel{(B-Ind)}
        \UnaryInfC{$(\neg \exists y. \neg d(w, y), \tau_w, \beta_w)$}
        \LeftLabel{(B-Cons)}
        \UnaryInfC{$\Pi^w_{\neg a(x)}$}
        \LeftLabel{(Proph)}
        \BinaryInfC{$\Pi$}
        \bottomAlignProof
    \end{prooftree}
    We prove the soundness problem $\Pi^\mathrm{sound}_{\neg a(x)}$ using an inductive invariant that says there is always an element satisfying both $\neg a(\cdot)$ and $m(\cdot)$, which means that adding a witness with axiom $\neg a(w)$ is sound. We then prove the problem $\Pi^w_{\neg a(x)}$ with the added witness via a backward inductive invariant that shows there is some element $y$ that is not in the relation $d(\cdot, \cdot)$ with $w$. This is indeed a backward inductive invariant since it is implied by $\beta$, and since backward transitions cannot change $\neg a(w)$ or $r(x, y)$ for any $x$ with $\neg a(x)$, backward consecution also holds. Lastly, if we denote this proof by $P$, we get $\FBPInv(P) = \exists x. \neg a(x) \wedge \neg \exists y. \neg d(x, y)$,  which is equivalent to the safe inductive invariant $\exists x. \neg a(x) \wedge \forall y. d(x,y)$ we considered before.
\end{example}

Finally, the overall results about the proof power of $\FBPIsys$ are summarized as follows.

\begin{theorem}
    \label{thm:FBP-power}
    For every $\predicates$ whose predicates have quantification depth of at most $d\in\nats$, all safety problems that are provable in $\FBPIsys_n^\predicates$ are also provable in $\Fsys^{\mathbb{Q}_{n\cdot d}}$, where $\mathbb{Q}_{n\cdot d}$ consists of all predicates with quantification depth at most $n\cdot d$.
    However, for every $n>0$ there exist $\predicates$ and a safety problem $\Pi$ such that $\Pi$ is provable in $\FBPIsys_n^\predicates$ but
    not provable in either $\FBPIsys_{n-1}^\predicates$, $\FBIsys^\predicates$, $\FIsys^\predicates$, or $\Fsys^\predicates$.
\end{theorem}

\begin{proof}
    The first part of the theorem is a direct corollary of \Cref{thm:FBP-invariant}.
    For the second part,
    define the vocabulary with one unary relation and $n-1$ binary relations $\Sigma = \{p_1(\cdot),r_2(\cdot,\cdot),\dots,r_n(\cdot,\cdot)\}$, and the set of predicates $\predicates = \{\exists x_1. p(x_1)\wedge m(x_1)\} \cup
    \{\exists x_i. r_i(w_i, x_i)\wedge m(x_i)\}_{i=2}^n$.
    Denote
    \[ \varphi = 
    \exists x_1. p_1(x_1) \wedge
    \forall x_2. \neg r_2(x_1,x_2) \vee
    \exists x_3. r_3(x_2,x_3) \wedge
    \forall x_4. \neg r_4(x_3,x_4) \vee \cdots\]
    and consider the safety problem $\Pi = (\varphi, \bot, \neg\varphi)$. This safety problem is provable in $\FBPIsys^\predicates_n$, by first adding a witness $w_2$ with axiom $p_1(w_2)$ in the forward direction, then a witness $w_3$ with axiom $r_2(w_2, w_3)$ in the backward direction, then a witness $w_4$ with axiom $r_3(w_3, w_4)$ in the forward direction, and so on.
    If we prove the soundness of the prophecies using the appropriate soundness invariants, the safe inductive invariant computed for this proof via $\FBPInv(\cdot)$ is exactly equivalent to $\varphi$.
    However, $\Pi$ is not provable in $\FBPIsys_{n-1}^\predicates$, $\FBIsys^\predicates$, $\FIsys^\predicates$, or $\Fsys^\predicates$, since any safe inductive invariant must be equivalent to $\varphi$, and there exists no such formula with quantification depth of less than $n$.
\end{proof}

\subsection{Prophecy Heuristics}
\label{sec:proph-heuristics}

When using the forward-backward proof system without prophecy (\Cref{def:fwd-bwd-inc-system}), one can approach the task of constructing a proof, either manually or automatically, by iteratively attempting to find meaningful forward- or backward-inductive invariants and adding them as axioms to the safety problem, then proceeding in the same fashion while exploring different proof depths. Proving the soundness of each step is straightforward since it involves a simple induction query.
However, when  the proof system also allows prophecy (\Cref{def:fwd-bwd-proph-system}),
the soundness premise of the (Proph) rule might require its own inductive invariant or a more complicated proof structure to prove. Therefore, coming up with an appropriate prophecy and proving its associated soundness problem introduces a new challenge in terms of proof search.

In this section we present a heuristic approach for narrowing the search space for prophecies and their soundness invariants. We propose to do so by recognizing patterns of certain sound prophecies, and providing explicit verification conditions for each pattern, which can be checked without explicitly constructing the soundness problem for the prophecy and attempting a direct proof. The idea is that, while our proof system is designed to allow general safety proofs, in practice we might want to restrict the kind of prophecies we use to simplify verifying their soundness.

For instance, consider the prophecy used in \Cref{ex:proph-quantifiers}. We can generalize this kind of prophecy to any formula $\varphi(x)$ which for every initial state holds on some element, i.e., $\iota \Rightarrow \exists x. \varphi(x)$, and is always forward-preserved, i.e., $\varphi(x)\wedge\tau \Rightarrow \varphi'(x)$. If these conditions hold, $\varphi(w)$ is obviously sound, and we would like to add it without exploring proof possibilities for its associated soundness problem. In other words, we would like to use the following inference rule directly:
\begin{prooftree}
\AxiomC{$(\iota,\tau,\beta)^w_{\varphi}$}
    \LeftLabel{(Proph-Fwd)}
    \RightLabel{$\iota \Rightarrow \exists x. \varphi(x)$,
    $\varphi(x) \wedge \tau \Rightarrow \varphi'(x)$}
    \UnaryInfC{$(\iota,\tau,\beta)$}
\end{prooftree}
And indeed, we can use this rule directly by deriving (Proph-Fwd) once for all prophecies of this kind, relying on the verification conditions above.
We derive (Proph-Fwd) for a problem $\Pi=(\iota,\tau,\beta)$, whose premise is $\Pi^w_{\varphi}$, by denoting $\Pi_{\varphi}^\mathrm{sound}=(\iota_m,\tau_m,\beta_m)$ and deriving:
\begin{prooftree}
        \AxiomC{}
        \LeftLabel{(Ind)}
        \UnaryInfC{$(\iota_m, \tau_m, \neg \exists x. \varphi(x)\wedge m(x))$}
        \LeftLabel{(Cons)}
        \UnaryInfC{$\Pi^\mathrm{sound}_{\varphi}$}
        \AxiomC{$\Pi^w_{\varphi}$}
        \LeftLabel{(Proph)}
        \BinaryInfC{$\Pi$}
        \bottomAlignProof
\end{prooftree}

Using this approach, different classes of prophecy can be described and their explicit verification conditions written down and verified. These can be used in proofs as prophecy short-cuts via the appropriate derived inference rules, e.g., (Proph-Fwd) for forward-preserved prophecy formulas. When searching for a proof, these can be used to constrain the kinds of prophecy considered, and to supply explicit verification conditions for each.

We present one more prophecy pattern which turns out to be central to our evaluation in \Cref{sec:eval}.
This prophecy considers a predicate $\theta$ that is forward-preserved, and thus partitions any trace of the program into an initial part where all states satisfy $\neg\theta$, and a tail where all elements satisfy $\theta$. This must be the case because once $\theta$ holds at some point in the trace it must hold thereafter, since it is forward-preserved. The added prophecy $\varphi(x)$ has the property that it holds for all elements before $\theta$ holds; that it is forward preserved for any individual element once $\theta$ holds; and that $\exists x. \varphi(x)$ is itself an inductive invariant.
Intuitively, we can think of $\varphi(x)$ as a property that initially holds for all elements of a state until $\theta$ is activated, upon which (at least) one element is ``selected'' and continues to satisfy $\varphi(x)$ the entire length of the trace. The verification conditions detailed above translate to the following inference rule:
\begin{prooftree}
\AxiomC{$(\iota,\tau,\beta)^w_{\varphi}$}
    \LeftLabel{(Proph-Select)}
    \RightLabel{\shortstack{$\iota\Rightarrow \exists x. \varphi(x)$, $\exists x. \varphi(x) \wedge\tau \Rightarrow \exists x. \varphi'(x)$ \\
    $\neg\theta \Rightarrow \forall x. \varphi(x)$,
    $\theta \wedge \tau \Rightarrow \theta'$,\\    $\varphi(x) \wedge \theta \wedge \tau \Rightarrow\varphi'(x)$}}
    \UnaryInfC{$(\iota,\tau,\beta)$}
    \bottomAlignProof
\end{prooftree}
which can be derived in $\FBPIsys$, using similar notation to before:
\begin{prooftree}
        \AxiomC{}
        \LeftLabel{(Ind)}
        \UnaryInfC{$(\iota_m, \tau_m, \neg ((\neg \theta \wedge \forall x. m(x)) \vee (\theta \wedge \exists x. \varphi(x)\wedge m(x))))$}
        \LeftLabel{(Cons)}
        \UnaryInfC{$\Pi^\mathrm{sound}_{\varphi}$}
        \AxiomC{$\Pi^w_{\varphi}$}
        \LeftLabel{(Proph)}
        \BinaryInfC{$\Pi$}
        \bottomAlignProof
\end{prooftree}

Lastly, note that for each of the two prophecy rules, we do not only get explicit verification conditions, but also an explicit way to construct safe inductive invariants from safety proofs that use them. For example, if we add a prophecy $\varphi(x)$ using the (Proph-Fwd) rule and get an invariant $\psi$ for the problem with the prophecy $\Pi^w_{\varphi}$, then a safe inductive invariant for $\Pi$ is $\exists x. \varphi(x) \wedge \psi[x/w]$. Similarly, if we use the (Proph-Select) rule for adding prophecy, the resulting safe inductive invariant shall be $(\neg \theta \wedge \forall x. \psi[x/w]) \vee (\theta \wedge \exists x. \varphi(x)\wedge \psi[x/w])$.

%% file: 6-eval.tex
\section{Case Study: Proofs of Paxos Variants and Raft}
\label{sec:case-study}

We now demonstrate that the proof techniques developed in \Cref{sec:incremental-proofs,sec:proph} lead to simpler
proofs of the Paxos consensus protocol~\cite{paxos,paxos-made-simple}, several of its variants, and Raft~\cite{raft}.
Both Paxos and Raft are central distributed consensus protocols that have drawn the attention of the verification community. %
Our incremental
proofs use simpler predicates compared to the standard proofs that use a single (forward) inductive invariant.
This simplification of the proofs could be advantageous for automated proof search, as it narrows the space of potential invariant formulas.
While our case study shows promising results for several complex consensus protocols, we acknowledge that there is no guarantee that similar simplifications will occur in other protocols. We hope that the evidence presented here will motivate the community to apply our framework to complex protocols beyond those considered here.
We begin with a detailed illustration using single-decree Paxos, and follow with results from more complex variants.
All proofs reported here are available in the paper's artifact~\cite{artifact} and included as examples in the \href{https://github.com/wilcoxjay/mypyvy}{\texttt{mypyvy} tool}~\cite{mypyvy}.

\subsection{Proof Simplification for Single-Decree Paxos}
\label{sec:paxos}

\input{fig-paxos}

We start our investigation with the first-order logic model of single-decree Paxos listed in \Cref{fig:paxos}.
This model is a slightly simplified version of the model from~\cite{DBLP:journals/pacmpl/PadonLSS17}.\footnote{We abstract away the ``1a'' or ``start round'' messages since they are not important for the proof. We also state the safety property in terms of fresh constants rather than universal quantification (a form of Herbrandization) and focus on the case of consistency between two different rounds, assuming $r_1 < r_2$. Interestingly, the Herbrandization leads to a proof in EPR of the first-order model of Paxos without requiring the rewrites developed in~\cite{DBLP:journals/pacmpl/PadonLSS17}.\label{footnote:simplified}}
The key safety property we wish to prove is that only a
single value can be decided. Paxos ensures this by operating in \emph{rounds},
where each node remembers its current round, and does not participate in lower rounds.
A decision requires a majority (quorum) of votes on a proposed value,
and proposing a value in a new round requires a quorum of nodes that join that round.
When a node joins a round, it reports its latest vote so far (i.e., that in the highest round), and the proposed value must correspond to the latest vote reported by a quorum. This ensures that a proposal at a higher round cannot conflict with a decision at a lower round, including decisions that happen after the proposal was made (using a different quorum). See~\cite{paxos-made-simple,DBLP:journals/pacmpl/PadonLSS17} for more details.

We now show how a forward-backward proof,
whose proof tree is given in \Cref{fig:paxos-proof},
can simplify the predicates needed in the proof of Paxos.
The key source of simplification is the invariant that states that if a value is proposed in a round, then no other value can be \emph{choosable} at a lower round. Choosable here means that it is possible that the value will become decided in the future. That is, that there is a quorum of nodes all of which have either voted for that value or can still vote for it, i.e., they have not yet joined a higher round. In the context of the model listed in \Cref{fig:paxos}, where we assume $r_1 < r_2$ and $v_1 \neq v_2$, this invariant is formalized as follows:
\begin{equation} \label{eq:choosable-simple}
\begin{split}
& \forall r:\sround. \;
r_1 < r \wedge \rproposal(r,v_2) \to \;
\\
&
\qquad
\forall q:\squorum  \exists n:\snode. \;
\rmember(n,q) \land \neg \rvote(n,r_1,v_1) \land \rcurrentround(n) > r_1.
\end{split}
\end{equation}
That is, the invariant states that if $v_2$ is proposed at any round greater than $r_1$, then $v_1$ is no longer choosable at $r_1$ (i.e., in every quorum there is a node that is at a round greater than $r_1$ and did not vote for $v_1$ at $r_1$).
Of all of the invariants in the proof of Paxos, this one has the most complex Boolean structure. Note that it is not a clause because of the alternation of
$\vee$ and $\wedge$.
Interestingly, a forward-backward proof can replace this invariant with simpler predicates. The key point is that in a forward-only proof, we must represent the disjunction between $v_2$ being proposed (at any round greater than $r_1$) and $v_1$ being choosable (at $r_1$). However, in a
forward-backward
proof we can eliminate the disjunction, since we know that in a bad state there must be a decision on $v_1$ at $r_1$, so it must be choosable at every point along the path from init to bad.

\input{paxos-proofs-fig}

Formally, we first argue by a forward
proof step that a decision is supported by a quorum of votes (as is done in the forward-only proof), focusing on $r_1$ and $v_1$:
\begin{equation}
\label{eq:quorum-of-decision-1}   
\rdecision(r_1,v_1) \to
\exists q:\squorum. \forall n:\snode. \;\rmember(n, q) \to \rvote(n,r_1,v_1).
\end{equation}
This step corresponds to the first application of (Inc) in \Cref{fig:paxos-proof}, starting from the root, where $\varphi_1$ is \cref{eq:quorum-of-decision-1} above. This proof rule results in a left-hand premise proved by forward induction, and a right-hand premise where \cref{eq:quorum-of-decision-1} is assumed as an axiom (see \Cref{def:fwd-inc-ind-system}).

Next,
to prove the right-hand premise,
we argue with a backward
proof step that on any path to the bad states $v_1$ must always be choosable at $r_1$:
\begin{equation} \label{eq:choosable-backward}
\exists q:\squorum  \forall n:\snode. \;
\rmember(n,q) \land \rcurrentround(n) > r_1 \to \rvote(n,r_1,v_1).
\end{equation}
This step corresponds to the application of (B-Inc) in \Cref{fig:paxos-proof}, where $\varphi_2$ is \cref{eq:choosable-backward} above, and the left-hand premise is proved via backward induction.
 
Finally, using a forward step we can show that
in the transition system resulting from assuming \cref{eq:quorum-of-decision-1,eq:choosable-backward},
$v_2$ cannot be proposed at any round greater than $r_1$:
\begin{equation} \label{eq:choosable-forward}
\forall r:\sround. \;
r_1 < r \to \neg \rproposal(r,v_2).
\end{equation}
This easily leads to the ultimate safety proof. Formally, this last step corresponds to proving the forward inductive invariant $\varphi_3$ in \Cref{fig:paxos-proof}, which consists of \cref{eq:choosable-forward} and a few supporting invariants.

All in all, we replaced the forward proof with a forward-backward proof which uses less complex auxiliary inductive invariants.
The Boolean structure of \cref{eq:quorum-of-decision-1,eq:choosable-backward,eq:choosable-forward}
is simpler than that of \cref{eq:choosable-simple}.
This 
suggests that it might be easier to find these invariants automatically,
because the search space for clausal invariants is smaller.

Using prophecy, the invariants can be simplified even further, eliminating the existential quantifier in \cref{eq:choosable-backward} and replacing it with a witness obtained for the property expressed in \cref{eq:quorum-of-decision-1}.
That is, using prophecy, the first forward step in the proof can be replaced by the introduction of
a witness to the quorum that ultimately leads to the decision on $v_1$.
Then, in the backward step, we can state that $v_1$ must remain choosable using the same quorum, which is expressed using a purely universally quantified formula. 
That is, the forward invariant gives us a witness that makes the backward invariant simpler.
We note that while the use of prophecy variables in general is not a new proof technique, the synergy between forward-backward reasoning and prophecy is new.

Formally,
the quorum that comes from \cref{eq:quorum-of-decision-1}
is
introduced as a prophecy witness using the rule (Proph-Select) (\Cref{sec:proph-heuristics})
with the prophecy
$\phi_1^q =  \forall n:\snode. \; \rdecision(r_1,v_1) \wedge \rmember(n, q) \to \rvote(n,r_1,v_1)$
and $\theta = \rdecision(r_1,v_1)$.
The proof tree using prophecy is given in \Cref{fig:paxos-proof-proph}.

\subsection{Application to Paxos Variants and Raft}
\label{sec:eval}

\input{tab-eval}

Similarly to the simplification from \Cref{sec:paxos}, we use the incremental proof systems developed in \Cref{sec:incremental-proofs,sec:proph} to prove the safety of various Paxos variants
and an abstract model of Raft.
We compare the relative complexity of 
user-specified predicates used in
each proof method in terms of Boolean structure, the number of quantifiers and quantifier alternations required, as well as the number of
prophecy symbols
used by each proof.

We consider several Paxos variants modeled in~\cite{DBLP:journals/pacmpl/PadonLSS17}: Paxos, Flexible Paxos, Multi-Paxos, and Fast Paxos, all modeled in the effectively propositional (EPR) fragment of first-order logic; as well as the Paxos FOL variant presented (in simplified form) in \Cref{sec:case-study}. 
These Paxos variants have complex inductive invariants that are highly challenging to find automatically. This challenge drove the development of several recent invariant inference techniques, e.g.,~\cite{weaken,pfolic3,DuoAI,Swiss}.
Some of these techniques are able to find the required invariants automatically,
but in some cases
the search involves long runtimes (on the order of hours),
due to
the complexity of the invariants and the resulting search space.

In addition to the standard variants of the original Paxos examples from~\cite{DBLP:journals/pacmpl/PadonLSS17}, we extend our example set with versions of these which differ in the arity of the $decision$ relation of the protocol. In most of the original examples, this relation is modeled to indicate both which value was decided upon, and in which round. We introduce additional versions where the $decision$ relation indicates only the decided value, as well as versions where it indicates the decided value, in which round the decision took place, and which quorum of nodes was responsible for the decision.
These alternative versions,
which differ slightly in the level of abstraction and the complexity of required invariants, allow us to demonstrate the different proof power of the proof systems considered.

As for
Raft,
we consider the version
modeled in EPR from~\cite{modularity-raft-pldi18}. The invariant for this EPR modeling of the protocol is significantly more complex than the ones for Paxos, and it has not been found automatically by the above-mentioned tools. Other works considered variants
or specific parts
of Raft and
automated
some of the proof~\cite{endive,basilisk}.

As a preprocessing step, we rewrite the safety property of all examples using fresh immutable constants rather than universal quantification, similarly to the example in~\Cref{sec:paxos}. 

The complexity of user-specified predicates in
the proof of each example for different proof systems
is detailed in \Cref{tab:eval}.
For Paxos examples,
the suffix of the name of each
example indicates the arity of the $decision$ relation.
In all proofs we consider predicates that are conjunctions of formulas with a certain number of quantifiers (Q), quantifier alternations (A), and Boolean connectives (B). We indicate whether the formulas are purely clausal (contain only disjunctions). Since the class of predicates we consider is closed under conjunction, all forward proofs require no incrementality.
State-of-the-art invariant inference tools, e.g., P-FOL-IC3~\cite{pfolic3} and DuoAI~\cite{DuoAI}, produce invariants of similar complexity to the described forward proofs for their benchmarks.
For forward-backward proofs, we consider a sequential proof structure which incrementally adds a predicate in each step which is either forward-inductive or backward-inductive (relative to the predicates added so far),
similarly to the proof in \Cref{fig:paxos-proof}.
For forward-backward proofs with prophecy, proofs follow a similar pattern, but each step may also be a prophecy step where a sound prophecy is added instead of an inductive invariant
using the (Proph-Select) heuristic presented in~\Cref{sec:proph-heuristics}, similarly to the proof in \Cref{fig:paxos-proof-proph}.
We indicate in parentheses how many prophecy symbols were added in each step (p). We use a straightforward generalization of the prophecy proof rule presented in \Cref{sec:general-prophecy} which allows adding multiple prophecy symbols in each step.

All proofs were manually written and mechanically checked using \texttt{mypyvy}~\cite{mypyvy},
which discharges logical queries to the \textsc{Z3}~\cite{z3} SMT solver.
To check the proofs, \texttt{mypyvy} was extended with syntax and a proof checker for forward-backward proofs with prophecy of the sequential nature described above.
This extension has been incorporated into the \href{https://github.com/wilcoxjay/mypyvy}{\texttt{mypyvy} tool}, and is also available in the paper's artifact~\cite{artifact}.
\Cref{tab:eval} presents the verification time for each considered proof, computed as the median of 15 runs.
All experiments were carried out on a \texttt{z1d.large} AWS machine, which has 2 cores and 16GB of RAM.
The timing results show that for this case study, writing proofs in our proposed proof system noticeably reduces verification times for the more difficult examples, and matches verification times for the examples that are verified relatively quickly (in less than a second).

As can be seen in~\Cref{tab:eval}, for most examples, a forward-backward proof greatly simplifies the Boolean structure of predicates considered, by enabling proofs to use only clausal formulas (rather than ones containing both conjunctions and disjunctions). As expected, there is no reduction in the quantifier alternations required in forward-backward proofs compared to forward proofs. Interestingly,
Raft and the Paxos variants that contain only values in the $decision$ relation 
do not seem to
have forward-backward proofs that simplify forward-only proofs, due to quantification that prevents decomposing non-clausal invariants.
However,
using
prophecy steps, we can avoid this quantification and
again
prove safety using only conjunctions of 
universally quantified
clauses.

On the other hand, in the round-value Paxos variants, forward-backward proofs simplify the Boolean structure of invariants; and adding prophecy allows the proof to avoid
all existential quantification in user-specified predicates.
In fact, in all forward-backward proofs in \Cref{tab:eval}, we only ever need to use prophecy in the forward direction, and only 
in the beginning
of each proof.
Note that since we use the (Proph-Select) heuristic, only $\theta$ is specified by the user, which is quantifier-free in our examples, leading to proofs that only involve universal, clausal predicates. However, the soundness invariant for the prophecy, which is automatically constructed from $\theta$, contains existential quantification and quantifier alternations.

Lastly, we observe that including the deciding quorum in the decision relation of these protocols simplifies even the forward invariants to require no alternation. However, these still include existential quantification, and the predicates in these proofs are non-clausal. Forward-backward proofs simplify these to purely universally quantified clauses, and consequently, these examples cannot be simplified further using prophecy.

%% file: fig-paxos.tex
\lstset{ %
  basicstyle=\footnotesize,
  xleftmargin=2em,
  breakatwhitespace=false,         %
  keepspaces=true,                 %
  keywordstyle=\bf,       %
  language=C,                 %
  otherkeywords={safety,var,require,module,individual,init,action,returns,assert,assume,instantiate,isolate,mixin,before,relation,function,sort,variable,axiom,then,constant,*,local},           %
  numbers=left,                    %
  numbersep=5pt,                   %
  numberstyle=\tiny,               %
  rulecolor=\color{black},         %
  tabsize=8,	                   %
   columns=fullflexible,
   mathescape=true,
}
\begin{figure}[t]
\noindent
\begin{tabular}{ll}
\begin{minipage}{.38\textwidth}
\begin{lstlisting}[basicstyle=\scriptsize,name=paxos]
sort $\snode$, $\squorum$, $\sround$, $\svalue$

relation $\leq$ : $\sround,\sround$
axiom $\theorytotalorder[\leq]$

relation $\rmember$ : $\snode,\squorum$
axiom $\forall q_1,q_2 : \squorum. \,  \exists n:\snode. \,$ 
                      $\rmember(n,q_1) \land \rmember(n, q_2)$

function $\rcurrentround$ : $\snode \to \sround$
relation $\ronebmaxvote$ : $\snode,\sround,\sround,\svalue$
relation $\ronebnovote$ : $\snode,\sround$
relation $\rproposal$ : $\sround,\svalue$
relation $\rvote$ : $\snode,\sround,\svalue$
relation $\rdecision$ : $\sround,\svalue$

init $\forall n,r_1,r_2,v. \; \neg\ronebmaxvote(n,r_1,r_2,v)$
init $\forall n,r. \; \neg\ronebnovote(n,r)$
init $\forall r,v. \; \neg\rproposal(r,v)$
init $\forall n,r,v. \; \neg\rvote(n,r,v)$
init $\forall n,r,v. \; \neg\rdecision(n,r,v)$

individual $r_1,r_2$ : $\sround$
individual $v_1,v_2$ : $\svalue$
axiom $r_1 < r_2 \wedge v_1 \neq v_2$
safety $\neg (\rdecision(r_1,v_1) \wedge \rdecision(r_2,v_2))$
\end{lstlisting}
\end{minipage} &
\hfill
\begin{minipage}{.56\textwidth}
\begin{lstlisting}[basicstyle=\scriptsize,firstnumber=auto,name=paxos]
action $\ajoinround(\text{n} : \snode ,\, \text{r} : \sround)$ {
  require $\rcurrentround(\text{n}) < \text{r}$
  if $\forall r,v. \; \neg \rvote(\text{n},r,v)$ then {
      $\ronebnovote(\text{n},\text{r})$ := true
  } else {
    var maxr, v := $\max \{ (r',v') \mid \rvote(\text{n},r',v') \}$
    $\ronebmaxvote(\text{n},\text{r},\text{maxr},\text{v})$ := true 
  }
  $\rcurrentround(\text{n})$ :=  $\text{r}$
}
action $\apropose(\text{r} : \sround ,\, \text{q} : \squorum)$ {
  require $\forall v. \; \neg\rproposal(\text{r},v) \wedge $
                     $\forall n. \; \rmember(n, \text{q}) \to \, $ 
                           $\ronebnovote(n, \text{r}) \vee \exists r',v'. \; \ronebmaxvote(n,\text{r},r',v')$
  var maxr, v :=                                       # v is arbitrary $\text{if}$ the set is empty
                $\max \{ (r',v') \mid \exists n. \; \rmember(n, \text{q}) \land \ronebmaxvote(n,\text{r},r',v') \}$
  $\rproposal(\text{r}, \text{v})$ := true $\label{line:propose-send}$
}
action $\acastvote(\text{n} : \snode ,\, \text{r} : \sround ,\, \text{v} : \svalue)$ {
  require $\rproposal(\text{r}, \text{v}) \wedge \text{r} = \rcurrentround(\text{n})$
  $\rvote(\text{n}, \text{r}, \text{v})$ := true
}
action $\adecide(\text{r}\!:\!\sround ,\, \text{v}\!:\!\svalue ,\, \text{q}\!:\!\squorum)$ {
  require $\forall n. \; \rmember(n, \text{q}) \to \rvote(n, \text{r}, \text{v})$
  $\rdecision(\text{r}, \text{v})$ := true
}
\end{lstlisting}
\end{minipage}
\end{tabular}
\caption{\label{fig:paxos} Model of Paxos consensus algorithm as a transition system in many-sorted first-order logic.}
\end{figure}

%% file: paxos-proofs-fig.tex
\begin{figure}[t]
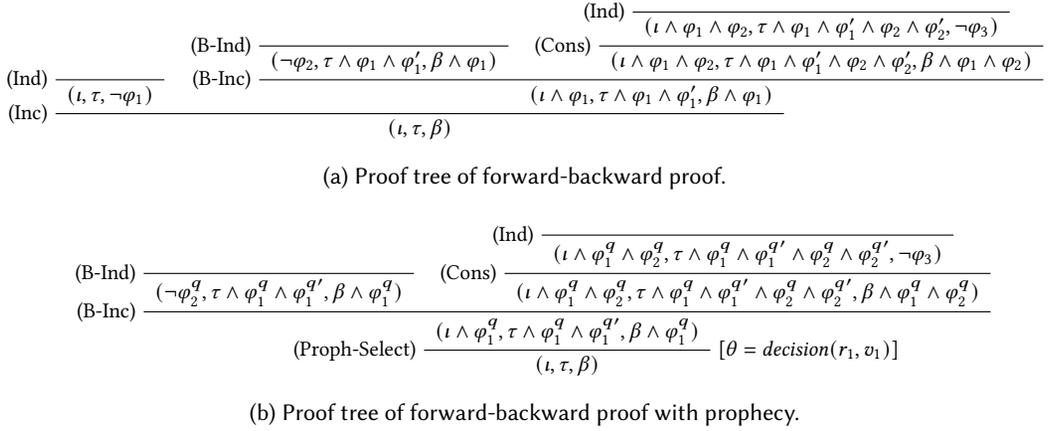

  \centering
  \begin{subfigure}{\linewidth}
    \centering
    \footnotesize
    \begin{prooftree}
        \def\defaultHypSeparation{\hskip 10pt}
        \AxiomC{}
        \LeftLabel{(Ind)}
        \UnaryInfC{$(\iota, \tau, \neg\varphi_1)$}
        \AxiomC{}
        \LeftLabel{(B-Ind)}
        \UnaryInfC{$(\neg\varphi_2, \tau\wedge\varphi_1\wedge\varphi_1', \beta\wedge\varphi_1)$}
        \AxiomC{}
        \LeftLabel{(Ind)}
        \UnaryInfC{$(\iota\wedge\varphi_1\wedge\varphi_2, \tau\wedge\varphi_1\wedge\varphi_1'\wedge\varphi_2\wedge\varphi_2', \neg \varphi_3)$}
        \LeftLabel{(Cons)}
        \UnaryInfC{$(\iota\wedge\varphi_1\wedge\varphi_2, \tau\wedge\varphi_1\wedge\varphi_1'\wedge\varphi_2\wedge\varphi_2', \beta\wedge\varphi_1\wedge\varphi_2)$}
        \LeftLabel{(B-Inc)}
        \BinaryInfC{$(\iota\wedge\varphi_1, \tau\wedge\varphi_1\wedge\varphi_1', \beta\wedge\varphi_1)$}
        \LeftLabel{(Inc)}
        \BinaryInfC{$(\iota, \tau, \beta)$}
        \bottomAlignProof
    \end{prooftree}
    \caption{Proof tree of forward-backward proof.}
    \label{fig:paxos-proof}
  \end{subfigure}\\[1em]
  \begin{subfigure}{\linewidth}
    \centering
    \footnotesize
    \begin{prooftree}
        \def\defaultHypSeparation{\hskip 10pt}
        \AxiomC{}
        \LeftLabel{(B-Ind)}
        \UnaryInfC{$(\neg\varphi_2^q,  \tau\wedge\varphi_1^q\wedge{\varphi_1^q}', \beta\wedge\varphi_1^q)$}
        \AxiomC{}
        \LeftLabel{(Ind)}
        \UnaryInfC{$(\iota\wedge\varphi_1^q\wedge\varphi_2^q, \tau \wedge\varphi_1^q\wedge{\varphi_1^q}' \wedge\varphi_2^q\wedge{\varphi_2^q}', \neg \varphi_3)$}
        \LeftLabel{(Cons)}
        \UnaryInfC{$(\iota\wedge\varphi_1^q\wedge\varphi_2^q, \tau \wedge\varphi_1^q\wedge{\varphi_1^q}' \wedge\varphi_2^q\wedge{\varphi_2^q}', \beta\wedge\varphi_1^q\wedge\varphi_2^q)$}
        \LeftLabel{(B-Inc)}
        \BinaryInfC{$(\iota\wedge\varphi_1^q, \tau\wedge\varphi_1^q\wedge{\varphi_1^q}', \beta\wedge\varphi_1^q)$}
        \LeftLabel{(Proph-Select)}
        \RightLabel{[$\theta=\rdecision(r_1,v_1)$]}
        \UnaryInfC{$(\iota, \tau, \beta)$}
        \bottomAlignProof
    \end{prooftree}
    \caption{Proof tree of forward-backward proof with prophecy.}
    \label{fig:paxos-proof-proph}
  \end{subfigure} 
  \caption{%
    Proof trees for incremental safety proofs of the Paxos protocol in \Cref{fig:paxos}, with side conditions omitted for simplicity. The triple $(\iota,\tau,\beta)$ denotes the safety problem defined by the protocol.
    Formulas $\varphi_1$ and $\varphi_2$ are \cref{eq:quorum-of-decision-1,eq:choosable-backward}, respectively, and $\varphi_3$ consists of \cref{eq:choosable-forward} and a few additional invariants which we omit to simplify the presentation.
    Formulas $\varphi^q_1$ and $\varphi^q_2$ are identical to \cref{eq:quorum-of-decision-1,eq:choosable-backward}, respectively, except $q$ is not quantified and is instead introduced as a fresh prophecy witness by (Proph-Select), which uses $\theta=\rdecision(r_1,v_1)$.}
  \label{fig:paxos-proofs}
\end{figure}

%% file: tab-eval.tex
\begin{table}[t]
    \caption{Parameters of user-specified predicates involved in the proof of various Paxos variants and Raft. Q = number of quantifiers, A = number of quantifier alternations, B = number of Boolean connectives, p = number of prophecy symbols. The class of predicates used consist of closing these under conjunction, so forward proofs are always of length 1. We indicate whether the predicates in the proof are clausal, i.e., include only disjunctions.
    For forward-backward proofs, we write the parameters of each iteration in a separate line, %
    denoting forward and backward steps by \textbf{F} and \textbf{B}, respectively, or \textbf{FP} and \textbf{BP}, respectively, when prophecy is used. %
    A left arrow ($\longleftarrow$) means there is no additional simplification so the proof is identical to the one on the left. The reported verification time for each proof technique is computed as the median runtime of 15 runs.}\label{tab:eval}
    \renewcommand{\tabcolsep}{2.85pt}%
    \centering%
    \scriptsize%
    {\linespread{0.85}\selectfont%
    \begin{tabular}{|c|c|c|c|c|c|c|c|c|c|}
    \hline
        \multirow{2}{*}{\textbf{Example}} & \multicolumn{3}{c|}{\textbf{Forward}} & \multicolumn{3}{c|}{\textbf{Forward-Backward}} & \multicolumn{3}{c|}{\textbf{Forward-Backward + Prophecy}} \\ %
        & Predicates & Clausal & Time & Predicates & Clausal & Time & Predicates & Clausal & Time \\ \hline
        paxos-epr-v & Q=4, A=1, B=5 & $\times$ & 0.59\,s &
        \multicolumn{3}{c|}{$\longleftarrow$} &
        \makecell[l]{\textbf{FP:} Q=1, A=0, B=2, p=2\\\textbf{FP:} Q=1, A=0, B=2, p=2\\\textbf{B:\phantom{P}} Q=1, A=0, B=2\\\textbf{F:\phantom{P}} Q=2, A=0, B=2} & $\checkmark$ & 0.43\,s \\ \hline
        paxos-epr-rv & Q=3, A=1, B=5 & $\times$ & 0.46\,s &
        \makecell[l]{\textbf{F:\phantom{P}} Q=2, A=1, B=2\\\textbf{B:\phantom{P}} Q=2, A=1, B=2\\\textbf{F:\phantom{P}} Q=2, A=0, B=2} & $\checkmark$ & 0.45\,s &
        \makecell[l]{\textbf{FP:} Q=1, A=0, B=2, p=1\\\textbf{FP:} Q=1, A=0, B=2, p=1\\\textbf{B:\phantom{P}} Q=1, A=0, B=2\\\textbf{F:\phantom{P}} Q=2, A=0, B=2} & $\checkmark$ & 0.45\,s \\ \hline
        paxos-epr-rvq & Q=2, A=0, B=5 & $\times$ & 0.38\,s &
        \makecell[l]{\textbf{F:\phantom{P}} Q=1, A=0, B=2\\\textbf{B:\phantom{P}} Q=1, A=0, B=2\\\textbf{F:\phantom{P}} Q=2, A=0, B=2} & $\checkmark$ & 0.4\,s &
        \multicolumn{3}{c|}{$\longleftarrow$} \\ \hline
        flexible-paxos-epr-v & Q=4, A=1, B=5 & $\times$ & 0.54\,s &
        \multicolumn{3}{c|}{$\longleftarrow$} &
        \makecell[l]{\textbf{FP:} Q=1, A=0, B=2, p=2\\\textbf{FP:} Q=1, A=0, B=2, p=2\\\textbf{B:\phantom{P}} Q=1, A=0, B=2\\\textbf{F:\phantom{P}} Q=2, A=0, B=2} & $\checkmark$ & 0.43\,s \\ \hline
        flexible-paxos-epr-rv & Q=3, A=1, B=5 & $\times$ & 0.45\,s &
        \makecell[l]{\textbf{F:\phantom{P}} Q=2, A=1, B=2\\\textbf{B:\phantom{P}} Q=2, A=1, B=2\\\textbf{F:\phantom{P}} Q=2, A=0, B=2} & $\checkmark$ & 0.45\,s &
        \makecell[l]{\textbf{FP:} Q=1, A=0, B=2, p=1\\\textbf{FP:} Q=1, A=0, B=2, p=1\\\textbf{B:\phantom{P}} Q=1, A=0, B=2\\\textbf{F:\phantom{P}} Q=2, A=0, B=2} & $\checkmark$ & 0.44\,s \\ \hline
        flexible-paxos-epr-rvq & Q=2, A=0, B=5 & $\times$ & 0.37\,s &
        \makecell[l]{\textbf{F:\phantom{P}} Q=1, A=0, B=2\\\textbf{B:\phantom{P}} Q=1, A=0, B=2\\\textbf{F:\phantom{P}} Q=2, A=0, B=2} & $\checkmark$ & 0.4\,s &
        \multicolumn{3}{c|}{$\longleftarrow$} \\ \hline
        multi-paxos-epr-nirv & Q=3, A=1, B=5 & $\times$ & 0.56\,s &
        \makecell[l]{\textbf{F:\phantom{P}} Q=2, A=1, B=2\\\textbf{B:\phantom{P}} Q=2, A=1, B=2\\\textbf{F:\phantom{P}} Q=3, A=0, B=2} & $\checkmark$ & 0.54\,s &
        \makecell[l]{\textbf{FP:} Q=1, A=0, B=2, p=1\\\textbf{FP:} Q=1, A=0, B=2, p=1\\\textbf{B:\phantom{P}} Q=1, A=0, B=2\\\textbf{F:\phantom{P}} Q=3, A=0, B=2} & $\checkmark$ & 0.5\,s \\ \hline
        fast-paxos-epr-v & Q=5, A=1, B=7 & $\times$ & 1.5\,s &
        \multicolumn{3}{c|}{$\longleftarrow$} &
        \makecell[l]{\textbf{FP:} Q=1, A=0, B=3, p=3\\\textbf{FP:} Q=1, A=0, B=3, p=3\\\textbf{B:\phantom{P}} Q=1, A=0, B=3\\\textbf{F:\phantom{P}} Q=2, A=0, B=3} & $\checkmark$ & 0.8\,s \\ \hline
        fast-paxos-epr-rv & Q=3, A=1, B=7 & $\times$ & 1.2\,s &
        \makecell[l]{\textbf{F:\phantom{P}} Q=2, A=1, B=3\\\textbf{B:\phantom{P}} Q=2, A=1, B=3\\\textbf{F:\phantom{P}} Q=2, A=0, B=3} & $\checkmark$ & 1\,s &
        \makecell[l]{\textbf{FP:} Q=1, A=0, B=3, p=2\\\textbf{FP:} Q=1, A=0, B=3, p=2\\\textbf{B:\phantom{P}} Q=1, A=0, B=3\\\textbf{F:\phantom{P}} Q=2, A=0, B=3} & $\checkmark$ & 0.81\,s \\ \hline
        paxos-fol-v & Q=5, A=1, B=5 & $\times$ & 1.6\,s &
        \multicolumn{3}{c|}{$\longleftarrow$} &
        \makecell[l]{\textbf{FP:} Q=1, A=0, B=2, p=2\\\textbf{FP:} Q=1, A=0, B=2, p=2\\\textbf{B:\phantom{P}} Q=1, A=0, B=2\\\textbf{F:\phantom{P}} Q=4, A=0, B=3} & $\checkmark$ & 0.62\,s \\ \hline
        paxos-fol-rv & Q=4, A=1, B=5 & $\times$ & 0.65\,s &
        \makecell[l]{\textbf{F:\phantom{P}} Q=2, A=1, B=2\\\textbf{B:\phantom{P}} Q=2, A=1, B=2\\\textbf{F:\phantom{P}} Q=4, A=0, B=3} & $\checkmark$ & 0.65\,s &
        \makecell[l]{\textbf{FP:} Q=1, A=0, B=2, p=1\\\textbf{FP:} Q=1, A=0, B=2, p=1\\\textbf{B:\phantom{P}} Q=1, A=0, B=2\\\textbf{F:\phantom{P}} Q=4, A=0, B=3} & $\checkmark$ & 0.6\,s \\ \hline
        paxos-fol-rvq & Q=4, A=0, B=5 & $\times$ & 0.52\,s &
        \makecell[l]{\textbf{F:\phantom{P}} Q=1, A=0, B=2\\\textbf{B:\phantom{P}} Q=1, A=0, B=2\\\textbf{F:\phantom{P}} Q=4, A=0, B=3} & $\checkmark$ & 0.54\,s &
        \multicolumn{3}{c|}{$\longleftarrow$} \\ \hline
        raft-epr & Q=7, A=2, B=9 & $\times$ & 3.7\,s &
        \multicolumn{3}{c|}{$\longleftarrow$} &
        \makecell[l]{\textbf{F:\phantom{P}} Q=7, A=0, B=8\\\textbf{FP:} Q=1, A=0, B=2, p=3\\\textbf{FP:} Q=1, A=0, B=2, p=3\\\textbf{B:\phantom{P}} Q=1, A=0, B=2\\\textbf{F:\phantom{P}} Q=3, A=0, B=5\\\textbf{F:\phantom{P}} Q=3, A=0, B=5} & $\checkmark$ & 2.8\,s \\ \hline
    \end{tabular}}%
\end{table}

%% file: 7-related.tex
\section{Related Work and Concluding Remarks}
\label{sec:related}

The idea of using a backward analysis to strengthen a forward analysis was first proposed theoretically by Patrick Cousot.\footnote{\emph{cf.} \cite{cousot:tel-00288657} pp. (3)-41,42.} It has since been applied in a variety of ways. For example, in symbolic model checking, a representation of the reachable states of a system, computed by a forward analysis, is typically used as a constraint on the backward fixed point iterations when evaluating future-time temporal operators~\cite{McMillanThesis}. This results in a simplification of the representation of the fixed point approximations, though it does not increase precision, since neither the forward nor the backward analyses are approximate. 

In~\cite{DBLP:conf/oopsla/CousotCLB12} iterated forward and backward analysis is used to strengthen contracts inferred for methods of sequential programs created by refactoring. In~\cite{DBLP:conf/tacas/VizelGS13} forward and backward interpolants are used for reachability analysis.
In~\cite{DBLP:conf/sas/BakhirkinM17} a forward (top-down) and backward (bottom-up) collecting semantics for Horn clauses are alternated to verify sequential programs that cannot be verified using either analysis alone, and it is found that multiple alternations are sometimes needed.  
A more recent method~\cite{DBLP:journals/corr/abs-2508-15137} uses an approximate backward analysis to strengthen the forward construction of an abstract reachability tree using Craig interpolants as in~\cite{Impact}.
None of these works address the problem of quantification, however.

By contrast, in this paper we consider the possible advantages of forward-backward analysis in the particular case of quantified invariants of parameterized protocols. Forward-backward analysis has not been considered in the protocol verification literature. It has been observed~\cite{pfolic3} that the lack of a conjunctive normal form greatly increases the search space for invariant formulas when existential quantifiers are needed. Forward-backward invariant generation can address this problem in two ways. First, in the alternation-free case, we find that a forward-backward approach can in practice reduce the space of invariants to only CNF formulas, despite the presence of existential quantifiers. Second, we find that the approach combines well with the addition of
prophecy
variables, since \emph{(a)} this transformation can eliminate quantifier alternations, and \emph{(b)} we can infer useful prophecy variables for forward analysis from backward analysis, and \emph{vice versa}.  

We can compare this approach to other methods of constraining the invariant search for protocol invariants. One approach is to bias the search toward certain forms that are hypothesized to occur commonly in practice. For example, in~\cite{pfolic3}, formulas of the form $\forall X.\ (C \rightarrow \exists Y. D)$, where $X,Y$ are variable sets and $C,D$ are cubes. This technique is effective for consensus protocols, since it captures the form of an important invariant, but as a heuristic it is possibly over-fitted to such protocols. By using forward-backward analysis, we can achieve a similar effect in a way that is more principled and possibly more general. Many other techniques that can synthesize mixed universal/existential invariants are also based on ad-hoc syntactic biases. These include \cite{DuoAI,DBLP:journals/corr/abs-2404-18048} which are based on reachable state sampling and~\cite{Swiss} which is based on brute-force enumeration of a template space with SMT-based checking. In principle, such methods could benefit from the reduction of formula complexity obtainable with a forward-backward approach. In other words, the forward-backward approach may allow a stronger syntactic bias. For the sampling methods, it would be necessary to sample backward reachable states, which might present computational challenges. For property-driven invariant generation methods, such as~\cite{folic3,pfolic3}, there is a different challenge, since these methods cannot generate an invariant in the absence of a property to be proved.  It is not clear what property should be proved in any forward or backward pass except the last. Still, there would be no difficulty in using a property-driven method in the last pass and a syntax-guided method in earlier passes.
In any event, the benefit of forward-backward analysis in this application is largely orthogonal to the method of analysis and has the potential to improve various techniques that apply syntactic biases to invariant generation. 

History and prophecy variables have very wide application in program verification. History variables were introduced by Clint~\cite{Clint} to solve a problem in using Hoare-style reasoning for coroutines. Prophecy variables were introduced by Abadi and Lamport~\cite{abadi1988the} to resolve non-determinism in proofs of refinement,
and are still widely used in that context~\cite{prophecy-made-simple}.
The idea that history and prophecy variables can be used to simplify the quantifier structure of invariant formulas may be considered folklore. Uncovering the first use of the idea is difficult but, for example, history and prophecy variables are used in~\cite{FlashVerify} to eliminate quantifiers from the invariant of a cache coherence protocol. History and prophecy are explicitly introduced as a method to eliminate quantifiers in~\cite{EagerAbstraction}. 
History variables are also explicitly used for this purpose in~\cite{DBLP:conf/pldi/TaubeLMPSSWW18} in the form of abstract modules. More recent work automatically introduces first-order prophecy variables to eliminate universal quantifiers ~\cite{prophic3,VickMcMillan}. 

What is primarily novel here, relative to the above-mentioned works is that we consider the interaction of history/prophecy with a forward-backward proof heuristic (note that history variables might be considered prophecy variables of the time-reversed system).
We also formalize the side condition for adding prophecy as an auxiliary safety problem.
We observe that existential invariants in one direction can be used as a guide to introduce prophecy variables in the other direction. Moreover, this transformation reduces the search space for forward or backward invariants by eliminating quantifier alternations. Thus, the two methods can be seen as mutually reinforcing.

%% file: 8-proof-appendix.tex
\section{Proofs}
\label{sec:appendix}

\subsection*{Proof of \Cref{thm:F-invariant}}

\begin{proof}
    The proof proceeds by induction on the structure of $P$. The cases where (Ind) or (Cons) is the root the proof are straightforward. If (Inc) is the root of $P$, then 
    $\varphi_1=\FInv(P_1)$ is a safe inductive invariant of a problem $(\iota,\tau,\neg\varphi)$,
    $\varphi_2=\FInv(P_2)$ is a safe inductive invariant of a problem $(\iota\wedge\varphi,\tau\wedge\varphi\wedge\varphi',\beta\wedge\varphi)$, and we claim $\varphi_1\wedge\varphi_2$ is a safe inductive invariant of $(\iota,\tau,\beta)$. 
    And indeed, initiation holds because $\iota\Rightarrow\varphi_1$, $\varphi_1\Rightarrow \varphi$, and $\iota\wedge\varphi \Rightarrow \varphi_2$, and therefore $\iota\Rightarrow\varphi_1\wedge\varphi_2$;
    consecution holds because $\tau\wedge\varphi_1\Rightarrow\varphi_1'$, $\varphi_1\Rightarrow \varphi$, and $\tau\wedge \varphi\wedge \varphi'\wedge \varphi_2\Rightarrow \varphi_2'$, and therefore $\tau \wedge \varphi_1\wedge\varphi_2 \Rightarrow (\varphi_1\wedge\varphi_2)'$; and safety holds because $\varphi_1\Rightarrow \varphi$ and $\varphi_2 \Rightarrow \neg(\beta\wedge\varphi)$, and therefore $\varphi_1\wedge\varphi_2 \Rightarrow \neg\beta$. Thus, $\FInv(P)$ is a safe inductive invariant for the conclusion of $P$.

    Lastly, from \Cref{def:F-invariant}, it is clear that $\FInv(P)$ is a conjunction of the predicates $\varphi\in P$ introduced in all (Ind) rules of the proof. Thus, if $P$ contains $n$ instances of the (Ind) rule, then $\FInv(P)$ is a conjunction of $n$ predicates from $\predicates$.
\end{proof}

\subsection*{Proof of \Cref{thm:FB-invariant}}

\begin{proof}
    The proof proceeds by induction on the structure of $P$. The cases where (Ind), (Cons) or (Inc) is the root of the proof are shown in the proof of \Cref{thm:F-invariant}. If (Rev) is the root of $P$, then we show that if $\varphi$ is a safe inductive invariant of $(\beta, \tau^{-1}, \iota)$, then $\neg\varphi$ is a safe inductive invariant of $(\iota,\tau,\beta)$. Initiation for $\neg\varphi$ holds because $\varphi \Rightarrow \neg \iota$ implies $\iota\Rightarrow\neg \varphi$; consecution holds because $\varphi\wedge\tau^{-1}\Rightarrow \varphi'$ implies $\neg\varphi \wedge\tau \Rightarrow \neg\varphi'$; and safety holds because $\varphi\Rightarrow \neg\beta$ implies $\neg\beta \Rightarrow \varphi$.

    Lastly, from \Cref{def:FB-invariant,def:F-invariant}, it is clear that $\FBInv(P)$ is a Boolean combination of the predicates $\varphi\in P$ introduced in all (Ind) rules of the proof. Thus, if $P$ contains $n$ instances of the (Ind) rule, then $\FBInv(P)$ is a Boolean combination of $n$ predicates from $\predicates$.
\end{proof} 

\subsection*{Proof of \Cref{thm:FBP-invariant}}

The correctness of the invariant computed in \Cref{def:FBP-invariant}, as stated in \Cref{thm:FBP-invariant}, is a direct consequence of \Cref{lem:prophecy-invariant-conversion} linking soundness invariants and invariants of the original problem.

\begin{lemma}
    \label{lem:prophecy-invariant-conversion}
    Let $\Pi=(\iota, \tau, \beta)$ be a transition system over vocabulary $\Sigma$, let $\varphi(x)$ be a formula over $\Sigma$, and $w$ a fresh constant symbol not in $\Sigma$. Let $\xi$ be a soundness invariant for $\varphi(w)$ (over $\Sigma\cup\{m(\cdot)\}$) and  $\psi$ a safe inductive invariant for $\Pi^{w}_{\varphi}$ (over $\Sigma\cup\{w\}$).
   Then $\xi[\psi/m]$, which is closed and over $\Sigma$, is a safe inductive invariant for $\Pi$.

\end{lemma}

\begin{proof}[Proof (\Cref{lem:prophecy-invariant-conversion})]
    For convenience, we write $\varphi(x)$ and $\psi(x)$ for the formulas over $\Sigma$ with free variable $x$, and $\varphi(w)$ and $\psi(w)$ for the closed formulas over $\Sigma\cup\{w\}$
    (in particular, $\psi$ in the lemma is $\psi(w)$).

    \textit{Initiation.} Let $\sigma\models\iota$ be an initial state. Extend it to $\Hat{\sigma}$ by interpreting $m$ as $m^{\Hat{\sigma}} = \psi^{\sigma}$.
    Due to initiation for $\psi(w)$ it holds that $\iota\wedge \varphi(w)\Rightarrow \psi(w)$, and so
    $\Hat{\sigma},x\mapsto a \models \iota\wedge \varphi(x)$ implies $\Hat{\sigma},x\mapsto a \models m(x)$ for all $a\in \Hat{\sigma}$. Therefore,
    $\Hat{\sigma}$ is an initial state of the soundness problem in \Cref{thm:sound-prophecy-as-safety}, and $\Hat{\sigma}\models \xi$ (initiation for $\xi$).
    Immediately from the definition of $\Hat{\sigma}$, we get $\Hat{\sigma}\models \xi[\psi/m]$.
    
    \textit{Consecution.} Let $(\sigma,\sigma') \models \tau$ be a transition with $\sigma\models \xi[\psi/m]$.
    Extend $\sigma$ and $\sigma'$ to $\Hat{\sigma}$ and $\Hat{\sigma}'$, respectively, by interpreting 
    $m^{\Hat{\sigma}} = \psi^{\sigma}$ and
    $m^{\Hat{\sigma}'} = \psi^{\sigma'}$, which means $\Hat{\sigma}\models \xi$.
    Due to consecution for $\psi(w)$ it holds that $\psi(w)\wedge \varphi(w) \wedge \tau \wedge w' = w \wedge (\varphi(w))'\Rightarrow (\psi(w))'$, which in particular implies that $\Hat{\sigma}, x\mapsto a \models m(x)\wedge \varphi(x)\wedge \varphi'(x)$ implies $\Hat{\sigma}', x\mapsto a \models m'(x)$ for all $a\in \Hat{\sigma}$.
    Therefore, $(\Hat{\sigma},\Hat{\sigma}')$ is a transition of the soundness problem in \Cref{thm:sound-prophecy-as-safety}, and thus $\Hat{\sigma}'\models \xi$ (consecution for $\xi$). Immediately from the definition of $\Hat{\sigma}'$, we get $\sigma'\models \xi[\psi/m]$.

    \textit{Safety.} Let $\sigma$ be a state such that $\sigma\models \xi[\psi/m]$, and as usual extend it to $\Hat{\sigma}$ by interpreting $m$ as $m^{\Hat{\sigma}} = \psi^{\sigma}$. Thus, $\Hat{\sigma} \models \xi$, and due to the safety of $\xi$ we get $\Hat{\sigma} \models \neg \beta\vee \exists x. \varphi(x) \wedge m(x)$. Consequently, there exists an extension of $\sigma$ to $\sigma_{w}$ that interprets $w$ such that $\sigma_{w} \models \neg \beta\vee (\varphi(w) \wedge \psi(w))$. Due to the safety of $\psi(w)$ we get $\sigma_{w}\models \neg\beta$, and since $\beta$ is over $\Sigma$, it also holds that $\sigma\models\neg\beta$.
\end{proof}